\documentclass[runningheads]{llncs}
\usepackage{graphicx}
\usepackage{amsmath}
\usepackage{amssymb}
\usepackage{enumerate}
\usepackage{paralist}
\usepackage{hyperref}
\usepackage{multirow}
\usepackage[utf8]{inputenc}
\usepackage{cite}

\usepackage[font=small, labelfont=bf]{caption}
\usepackage[labelfont=default]{subcaption}
\usepackage{xspace}
\usepackage{wrapfig}
\usepackage{booktabs}
\usepackage{todonotes}
\usepackage{comment}
\usepackage[pagewise]{lineno}
\newcommand{\ver}{arxiv}
\newcommand{\arxapp}[2]{\ifthenelse{\equal{\ver}{conf}}{#2}{#1}}
\newcommand{\planarOrPlane}{plane}
\newcommand{\improvement}[2]{\ifthenelse{\equal{\planarOrPlane}{plane}}{#2}{#1}}

\usepackage{thm-restate}

\title{Orthogonal and Smooth Orthogonal Layouts of 1-Planar Graphs
  with Low Edge Complexity\thanks{This work started at
    Dagstuhl seminar 16452 ``Beyond-Planar Graphs: Algorithmics and
    Combinatorics''.  We thank the organizers and the other participants.}}

\titlerunning{Orthogonal and Smooth Orthogonal Layouts of 1-Planar Graphs}

\author{Evmorfia~Argyriou\inst{1} \and 
Sabine~Cornelsen\inst{2} \and 
Henry~F\"orster\inst{3} \and 
Michael~Kaufmann\inst{3} \and Martin~N\"ollenburg\inst{4}\orcidID{0000-0003-0454-3937} \and Yoshio~Okamoto\inst{5}\orcidID{0000-0002-9826-7074} \and 
Chrysanthi~Raftopoulou\inst{6} \and
Alexander~Wolff\inst{7}\orcidID{0000-0001-5872-718X}}

\authorrunning{E.~Argyriou et al.}
  
\institute{yWorks GmbH, T\"ubingen, Germany
\email{evmorfia.argyriou@yworks.com} \and
University of Konstanz, Germany 
\email{sabine.cornelsen@uni-konstanz.de} \and
University of T\"ubingen, Germany 
\email{\{foersth,mk\}@informatik.uni-tuebingen.de} \and
TU Wien, Vienna, Austria \email{noellenburg@ac.tuwien.ac.at} \and
University of Electro-Communications and RIKEN Center for
Advanced Intelligence Project, Ch\=ofu, Japan \email{okamotoy@uec.ac.jp } \and 
National Technical University of Athens, Greece
\email{crisraft@mail.ntua.gr} \and
University of W\"urzburg, W\"urzburg, Germany \\
}

\graphicspath{{figures/}}

\let\doendproof\endproof
\renewcommand\endproof{~\hfill$\qed$\doendproof}

\newcommand{\SC}[1]{SC$_{#1}$}
\newcommand{\OC}[1]{OC$_{#1}$}

\begin{document}

\maketitle

\begin{abstract}
  While \emph{orthogonal} drawings have a long history, \emph{smooth
    orthogonal} drawings have been introduced only recently.  So far,
  only planar drawings or drawings with an arbitrary number of
  crossings per edge have been studied.  Recently, a lot of research
  effort in graph drawing has been directed towards the study of
  beyond-planar graphs such as \emph{1-planar} graphs, which admit a
  drawing where each edge is crossed at most once.
  \improvement{%
    In this paper, we present algorithms that yield (smooth)
    orthogonal drawings of 1-planar graphs with small curve complexity
    such that the given embedding is preserved. The curve complexity
    improves if we can modify the embedding.  For the subclass of
    outer-1-planar graphs, which can be drawn such that 
    all vertices lie on the outer face, we achieve curve complexities
    that are optimal for this class and fixed embedding.}%
  {In this paper, we consider graphs with a fixed embedding. For
    1-planar graphs, we present algorithms that yield orthogonal
    drawings with optimal curve complexity and smooth orthogonal
    drawings with small curve complexity. For the subclass of
    outer-1-planar graphs, which can be drawn such that 
    all vertices lie on the outer face, we achieve optimal curve
    complexity for both, orthogonal and smooth orthogonal drawings.}
\end{abstract}

\section{Introduction}

Orthogonal drawings date back to the 1980's, with
Valiant's~\cite{Val81}, Leiserson's~\cite{Lei80} and
Leighton's~\cite{Lei84} work on VLSI layouts and floor-planning
applications and have been extensively studied over the years.  The
quality of an orthogonal drawing can be judged based on several
aesthetic criteria such as the required area, the total edge length,
the total number of bends, or the maximum number of bends per edge.
While schematic drawings such as orthogonal layouts are very popular
for technical applications (such as UML diagrams) still to date,  from
a cognitive point of view, schematic drawings in other applications
like subway maps seem to have disadvantages over subway maps drawn
with smooth B\'ezier curves, for example, in the context of path
finding~\cite{rnlhh-ovsmp-IJHCS13}.  In order to ``smoothen''
orthogonal drawings and to improve their readability, Bekos et
al.~\cite{BKKS13} introduced \emph{smooth orthogonal drawings} that
combine the clarity of orthogonal layouts with the artistic style of
Lombardi drawings~\cite{DEGKN12} by replacing sequences of ``hard''
bends in the orthogonal drawing of the edges by (potentially shorter)
sequences of ``smooth'' inflection points connecting circular arcs.
Formally, our drawings map vertices to points in~$\mathbb R^2$ and
edges to curves of one of the following two~types.

\begin{description}
\item[Orthogonal Layout:] Each edge is drawn as a sequence of vertical
  and horizontal line segments. Two consecutive segments of an edge
  meet in a bend.
\item[Smooth Orthogonal Layout~\cite{BKKS13}:] Each edge is drawn as a
  sequence of vertical and horizontal line segments as well as
  circular arcs: quarter arcs, semicircles, and three-quarter
  arcs. Consecutive segments must have a common tangent.
\end{description}

The maximum vertex degree is usually restricted to four since every vertex has
four available \emph{ports} (North, South, East, West), where the
edges enter and leave a vertex with horizontal or vertical
tangents. In addition, the usual model insists that no two edges
incident to the same vertex can use the same port.  Throughout this
paper, we restrict ourselves to graphs of maximum degree~four.

The \emph{curve complexity} of a drawing is the maximum number of
segments used for an edge. An \emph{\OC{k}-layout} is an orthogonal
layout with curve complexity~$k$, that is, an orthogonal layout with
at most $k-1$ bends per edge.  An \emph{\SC{k}-layout} is a smooth
orthogonal layout with curve complexity~$k$. 
For results, see Table~\ref{table:res}.

\begin{wrapfigure}[9]{r}{.45\linewidth}
  \vspace*{-6.3ex} 

\begin{subfigure}[b]{.48\linewidth}
    \centering
    \includegraphics[page=3]{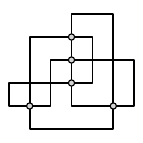}
    \caption{\OC{3}-layout}
    \label{fig:K5-oc3}
  \end{subfigure}
    \hfill
  \begin{subfigure}[b]{.48\linewidth}
    \centering
    \includegraphics[page=4]{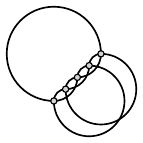}
    \caption{\SC{1}-layout}
    \label{fig:K5-sc1}
  \end{subfigure}

  \caption{Two 2-planar drawings of~$K_5$.}
  \label{fig:K5}
\end{wrapfigure}

The well-known algorithm of Biedl and Kant~\cite{BK98} draws any
connected graph of maximum degree~4 orthogonally on a grid of size $n
\times n$ with at most $2n + 2$ bends, bending each edge at most twice
(and, hence, yielding \OC3-layouts).   For the output of their algorithm applied
to~$K_5$, see Fig.~\ref{fig:K5-oc3}. Note that their approach
introduces crossings to the produced drawing.  For planar graphs, they describe how to obtain planar orthogonal
drawings with at most two bends per edge, except possibly for one edge
on the outer face.  

So far, smooth orthogonal drawings have been studied nearly
exclusively for planar graphs. Bekos et al.~\cite{BGPR14} showed how to compute an \SC1-layout for any
maximum degree~4 graph, but their algorithm does not
consider the embedding of the given graph.  For a drawing of~$K_5$
computed by their algorithm, see Fig.~\ref{fig:K5-sc1}.  Also, in the
produced drawings, the number of crossings that an edge may have is
not bounded.  Bekos et al.\ also showed that, if one does not restrict
vertex degrees, many planar graphs do not admit (planar) \SC1-layouts
under the \emph{Kandinsky model}, where the number of edges using the
same port is unbounded.  They proved, however, that all planar graphs
of maximum degree~3 admit an \SC1-layout (under the usual port
constraint).  For the same class of graphs, Alam et
al.~\cite{ABKKKW14} showed how to get a polynomial drawing area
($O(n^2) \times O(n)$) when increasing the curve complexity to \SC2.
Further, they showed that every planar graph of maximum degree~4
admits an \SC2-layout, but not every such graph admits an \SC1-layout
where the vertices lie on a polynomial-sized grid.  They also proved
that every biconnected outerplane graph of maximum degree~4 admits an
\SC1-layout (respecting the given embedding).

In this paper, we study orthogonal and smooth orthogonal layouts of
non-planar graphs, in particular, 1-planar graphs.  Recall that
$k$-planar graphs are those graphs that admit a drawing in the plane
where each edge has at most~$k$ crossings.  Our goal is to extend the
well-established aesthetic criterion `curve complexity' of (smooth)
orthogonal drawings from planar to 1-planar graphs.

1-planar graphs, introduced by Ringel~\cite{Ringel1965},
probably form the most-studied class of the 
\emph{beyond-planar} graphs, which extend the notion of planarity.
There are recent surveys on both 1-planar graphs~\cite{KobourovLM2017}
and beyond-planar graphs~\cite{DidimoLM2018}. 
Mostly, straight-line drawings have been studied for 1-planar graphs. 
While every planar graph has a planar straight-line drawing
(due to F\'ary's theorem), this is not true for 1-planar
graphs~\cite{eggleton:86,Thomassen1988}.  For the 3-connected case,
the statement holds except for at most one edge on the outer
face~\cite{abk-slgd3-GD13}.  Given a drawing of a 1-planar graph, one
can decide in linear time whether it can be ``straightened''
\cite{help-ft1pg-COCOON12}.

An important subclass of 1-planar graphs are \emph{outer-1-planar}
graphs.  These are the graphs that have a 1-planar drawing where every
vertex lies on the outer (unbounded) face.  They are planar graphs,
can be recognized in linear time~\cite{Hong2015,abbghnr-01pg-Alg16},
and can be drawn with straight-line edges and right-angle
crossings~\cite{dehkordi/eades:12}.

We are specifically interested in 1-plane and outer-1-plane graphs,
which are 1-planar and outer-1-planar graphs together with an
embedding.  Such an embedding determines 
the order of the edges around each vertex, but also which
edges cross and in which order.  
By the \emph{layout of a 1-plane graph} we 
mean that the layout respects the given embedding, without 
stating this again. In contrast, the \emph{layout of a 1-planar graph}
can have any 1-planar embedding.

\paragraph{Our contribution.}

Previous results and our contribution on (smooth) orthogonal 
layouts are listed in Table~\ref{table:res}.  We present new
layout algorithms for $1$-planar graphs in the orthogonal model
(Section~\ref{sec:ortho}) and in the smooth orthogonal model
(Section~\ref{sec:smooth}), achieving low curve complexity and
preserving 1-planarity.  We study 
1-plane graphs
as well as the special case of outer-1-plane graphs, where all
vertices lie on the outer face.  We conclude with some open problems; see Section~\ref{sec:future}.

In particular, we show that all 1-plane graphs admit 
\improvement{%
\OC5-layouts
(Theorem~\ref{thm:biconnected}) and, if they are biconnected,
\SC4-layouts (Theorem~\ref{thm:biconnected1plane-SC4}).  By altering
the initial embeddings, we can improve the curve complexity by one in
both models (Theorems~\ref{thm:1-planar_OC4}
and~\ref{thm:biconnected1planar-SC3}).}%
{\OC4-layouts (Theorem~\ref{thm:1-planar_OC4}) and \SC3-layouts
  (Theorem~\ref{thm:biconnected1planar-SC3}).}
  We also prove that all
biconnected outer-1-plane graphs admit \OC3-layouts
(Theorem~\ref{thm:outer1plane-OC3}) and \SC2-layouts
(Theorem~\ref{thm:biconnectedouter1plane-SC2}).  \improvement{The last two results}{Three out of these four results}
are worst-case optimal: There exist \improvement{}{biconnected 1-plane graphs that do not admit an \OC3-layout (Theorem~\ref{thm:biconnected1plane-not-OC3}) and} biconnected outer-1-plane graphs
that do not admit \OC2-layouts (Theorem~\ref{thm:outer1plane-not-OC2})
and \SC1-layouts (Theorem~\ref{thm:biconnectedouter1plane-not-SC1}).

\begin{table}[tb]
  \caption{Comparison of our results to previous work.  The model
    K(andinsky)-\SC1 does not restrict the number of edges per port to
    one.  ($^\star$) except for the octahedron (\OC4).  
    ``Super-poly'' means that the drawings are not known to be of
    polynomial size.}
  \label{table:res}
  \smallskip

  \centering
  \begin{tabular}{llcccc}
    \toprule
    &             & max. & curve       & drawing &  \\
    & graph class & deg. & complexity & area    & reference \\
    \midrule

    \parbox[t]{4mm}{\multirow{4}{*}{\rotatebox[origin=c]{90}{orthogonal}}} &
    general & 4 & \OC3 & $n \times n$ & \cite{BK98} \\
    & planar  & 4 & \OC3 ($^\star$) & $n \times n$ & \cite{BK98} \\[.75ex]
    & 1-plane & 4 & $\not\subseteq$ \OC3 / \improvement{\OC5}{\OC4} & $O(n) \times O(n)$ 
    & Thm.~\ref{thm:biconnected1plane-not-OC3} / \improvement{\ref{thm:biconnected}}{\ref{thm:1-planar_OC4}} \\
    \improvement{& 1-planar & 4 & \OC4 & $O(n) \times O(n)$ 
    & Thm.~\ref{thm:1-planar_OC4} \\}{}
    & biconnected outer-1-plane & 4 & $\not\subseteq$ \OC2 / \OC3 
    & $O(n) \times O(n)$  &
    Thm.~\ref{thm:outer1plane-not-OC2} / \ref{thm:outer1plane-OC3} \\

    \midrule
	
    \parbox[t]{4mm}{\multirow{10}{*}{\rotatebox[origin=c]{90}{smooth orthogonal}}} &
    planar & 4 & \SC2 & super-poly & \cite{ABKKKW14} \\
    & planar, poly-area & 4 & $\not\supseteq$ \SC1 & --- & 
    \cite{ABKKKW14} \\
    & planar, \OC2 & 4 & $\not\subseteq$ \SC1 & --- & \cite{ABKKKW14}\\
    & planar & 3 & \SC2 & $\lfloor {n^2}/4 \rfloor \times \lfloor n/2
    \rfloor$ & \cite{ABKKKW14} \\
    & planar & 3 & \SC1 & super-poly & \cite{BGPR14} \\
    & biconnected outerplane & 4 & \SC1 & super-poly &
    \cite{ABKKKW14} \\
    & general (non-planar) & 4 & \SC1 & $2n \times 2n$ & \cite{BGPR14} \\
    & planar & $\infty$ & $\not\subseteq$ K-\SC1 / K-\SC2 & $O(n) \times O(n)$ & \cite{BKKS13,BGPR14} \\[.75ex]
    & biconnected 1-plane & 4 & \improvement{\SC4}{\SC3} & \improvement{super-poly}{$O(n) \times O(n^2)$}
    & \improvement{Thm.~\ref{thm:biconnected1plane-SC4}}{Thm.~\ref{thm:biconnected1planar-SC3}} \\
    \improvement{& biconnected 1-planar & 4 & \SC3 & $O(n) \times O(n^2)$ 
    & Thm.~\ref{thm:biconnected1planar-SC3} \\}{}
    & biconnected outer-1-plane & 4 & $\not\subseteq$ \SC1 / \SC2 &
    super-poly
    & Thm.~\ref{thm:biconnectedouter1plane-not-SC1} /
    \ref{thm:biconnectedouter1plane-SC2} \\
	 
    \bottomrule
  \end{tabular}
\end{table}

\section{1-Planar Bar Visibility
  Representation}\label{sec:bar-visibility}

As an intermediate step towards orthogonal drawings, 
we introduce \emph{1-planar bar visibility
  representations}: Each vertex is represented as a horizontal
segment~--~called bar~--~and each edge is represented as either a
vertical segment or a polyline composed of a vertical segment and a
horizontal segment between the bars of its adjacent vertices. Edges
must not intersect other bars. If an edge has a horizontal segment, we
call it \emph{red}. The horizontal segment of a red edge must be on
top of its vertical segment and crosses exactly one vertical segment
of another edge~--~which is called \emph{blue}. The vertical segment
of a red edge must not be crossed; see
Fig.~\ref{fig:bar1planar_kites}.   We consider every edge
  as a pair of two \emph{half-edges}, one for each of its two 
  endpoints. Red edges are split at their bend~--~the \emph{construction bend}, such that
  each half-edge consists of either a vertical or a horizontal
  segment. Observe that horizontal half-edges are always red. 
 We show that every 1-planar graph has
a 1-planar bar visibility representation, following the approach of
Brandenburg~\cite{DBLP:journals/jgaa/Brandenburg14}:

For a 1-planar embedding, we define a \emph{kite} to be a $K_4$ induced by the
end vertices of two crossing edges with the property that each of the
four triangles induced by the crossing point and one end vertex of
each of the two crossing edges is a face. A crossing is \emph{caged}
if its end vertices induce a kite. Let now $G$ be a 1-planar graph. As
a preprocessing step, $G$ is augmented to a not necessarily simple
graph $G'$, with the property that any crossing is caged and no planar
edge can be added to $G'$ without creating a new crossing or a double
edge~\cite{abk-slgd3-GD13}.  
\improvement{Note that during this process the initial
embedding of $G$ might be altered.}{}
 
After the preprocessing step, all crossing edges are removed and a bar
visibility representation for the produced plane graph $G_p$ is
computed~\cite{DBLP:journals/dcg/RosenstiehlT86,DBLP:journals/dcg/TamassiaT86}. To this end 
an $st$-ordering of a biconnected supergraph of $G_p$ is computed, i.e., an
ordering $s=v_0$, $v_1$, $\ldots$, $v_{n-2}$, $v_{n-1}=t$ of the
vertices such that each vertex except $s$ and $t$ is adjacent to both,
a vertex with a greater and a lower index. The $st$-number is the
index of a vertex.  The y-coordinate of each bar is chosen to be the
$st$-number of the respective vertex.

Faces of size four that correspond to the kites of $G$ have three
possible configurations: left/right wing or diamond configuration.
Fig.~\ref{fig:bar1planar_kites} shows the configurations and how to
insert the crossing edges in order to obtain a 1-planar bar visibility
representation of~$G'$. Removing the caging edges results in a
1-planar bar visibility representation of~$G$. 

\begin{figure}[tb]
  \centering
  \begin{subfigure}[b]{.25\linewidth}
    \centering
    \includegraphics[page=1]{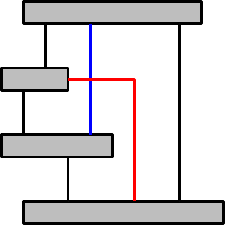}
    \caption{left wing}
    \label{fig:bar1planar_kites_left}
  \end{subfigure}
  \qquad
  \begin{subfigure}[b]{.25\linewidth}
    \centering
    \includegraphics[page=2]{bar1planarK4}
    \caption{right wing}
    \label{fig:bar1planar_kites_right}
  \end{subfigure}
  \qquad
  \begin{subfigure}[b]{.25\linewidth}
    \centering
    \includegraphics[page=3]{bar1planarK4}
    \caption{diamond}
    \label{fig:bar1planar_kites_diamond}
  \end{subfigure}
  \caption{Different configurations for kites in a 1-planar bar
    visibility representation.}
  \label{fig:bar1planar_kites}
\end{figure}

An edge is a \emph{left}, \emph{right}, \emph{top} or \emph{bottom
  edge} for a bar if it is attached to the respective side of that
bar.  Note that only red edges of $G$ can be left or right edges for
exactly one of their endpoints (and top edge for their other
endpoint). If a bar has no bottom (top) edges, it is a \emph{bottom}
(\emph{top}) bar, respectively. Otherwise it is a \emph{middle bar}.
For a bottom (top) bar, consider the x-coordinates of the touching
points of its edges. We define its \emph{leftmost} and \emph{rightmost
  edge} to be the edge with the smallest and largest x-coordinate,
respectively. If such a bar has a left or right edge then, by the
previous definition, this is its leftmost or rightmost edge,
respectively. Note that by the construction of the bar visibility
representation, each bar has at most one left and at most one right
red edge.

\section{Orthogonal 1-Planar Drawings}
\label{sec:ortho}

In this section, we examine orthogonal 1-planar drawings. In
particular, we give a counterexample showing that not every
biconnected 1-plane graph of maximum degree~4 admits an
\OC{3}-layout. On the other hand, we prove that every\improvement{ 
  biconnected}{} 1-plane graph of maximum degree~4 admits an
\improvement{\OC{5}}{\OC{4}}-layout that preserves the given
embedding\improvement{, and an \OC{4}-layout if we may change the
  embedding}{}.  \improvement{Additionally, for}{For} biconnected
outer-1-plane graphs we achieve optimal curve complexity~3.

\subsection{Orthogonal Drawings for General 1-Planar Graphs}

\begin{theorem}
  \label{thm:biconnected1plane-not-OC3}
  Not every biconnected $1$-plane graph of maximum degree~$4$ admits
  an \OC{3}-layout. Moreover, there is a family of
  graphs requiring a linear number of edges of complexity at least~$4$ in any \OC{4}-layout respecting the embedding.
\end{theorem}
\begin{proof}
  Consider the 1-planar embedding of a $K_5$ as shown in
  Fig.~\ref{fig:K5-oc4}. The outer face is a triangle $T$ and all
  vertices have their free ports in the interior of $T$.  Hence, $T$
  has at least $7$ bends, and at least one edge of $T$ has at least
  $3$ bends.
 
  For another example refer to Fig.~\ref{fig:no-oc3}, where vertices
  $a$, $b$, and $c$ create a triangle with the same properties.  We
  use $t$ copies of the graph of Fig.~\ref{fig:no-oc3} in a column and
  glue them together by connecting the top and bottom gray vertices of
  consecutive copies with an edge, as well as the
  topmost vertex of the topmost copy and the bottommost vertex of the
  bottommost copy. The graph has $n=9t$ vertices and at least $t$
  edges of complexity at least~4.
\end{proof}

\begin{figure}[tb]
  \begin{minipage}[b]{.65\linewidth}
    \centering
    \begin{subfigure}[b]{.5\linewidth}
      \centering
      \includegraphics[page=6]{K5}
      \caption{$K_5$}
      \label{fig:K5-oc4}
    \end{subfigure}
    \hfill
    \begin{subfigure}[b]{.4\linewidth}
      \centering
      \includegraphics{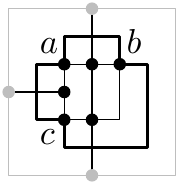}
      \caption{a 9-vertex graph}
      \label{fig:no-oc3}
    \end{subfigure}
    \caption{Biconnected 1-plane graphs without \OC{3}-layout}
  \end{minipage}
  \hfill
  \begin{minipage}[b]{.23\linewidth}
    \centering
    \includegraphics[page=2]{octahedron-BK}
    \vspace{8mm}
    \caption{Octahedron} 
    \label{FIG:octahedron}
  \end{minipage}
\end{figure}

\improvement{\begin{theorem}
  \label{thm:biconnected}
  Every $1$-plane graph of maximum degree~$4$ admits an \OC5-layout in an
  $n' \times n'$ grid, where $n'$ is the number of vertices plus the
  number of crossings.
\end{theorem}

\begin{proof}
  Let $G$ be a 1-plane graph of maximum degree~4, and let $G_p$ be the
  planarization of~$G$ obtained by using dummy vertices at
  crossings. If $G_p$ is not the octahedron, the algorithm of Biedl
  and Kant~\cite{BK98} yields a drawing of $G$ where uncrossed edges
  have at most two bends and crossed edges have at most four bends.  If $G_p$ is
  the octahedron, we use the orthogonal drawing from Fig.~\ref{FIG:octahedron}. No matter which vertices are
  dummy, all edges of $G$ bend at most three times.
\end{proof}

By altering the initial embedding of the graph, we can guarantee better curve complexity, as stated in Theorem~\ref{thm:1-planar_OC4}
below. We}{In order to achieve an \OC4-layout for 1-plane graphs, we} will use a general property of orthogonal drawings
of planar graphs: Consider two consecutive bends on an edge $e$ with
an incident face $f$. We say that the pair of bends forms a
\emph{U-shape} if they are both convex or both concave in $f$ and an
\emph{S-shape}, otherwise. It follows from the flow model of
Tamassia~\cite{tamassia87} that if a planar graph has an orthogonal
drawing with an S-shape then it also has an orthogonal drawing with
the identical sequence of bends on all edges except for the two bends
of the S-shape that are removed. Thus, by planarization, any pair of
S-shape bends can be removed as long as the two bends are not
separated by crossings.

\begin{theorem}
  \label{thm:1-planar_OC4}
  Every $n$-vertex $1$-\improvement{planar}{plane} graph of maximum degree~$4$ admits an
  \OC4-layout on a grid of size $ O(n) \times O(n)$.
\end{theorem}

\begin{proof}
  Let $G$ be a 1-planar graph of maximum degree~4 and consider a
  1-planar bar visibility representation of $G$.  If $G$ is not
  connected, we draw each connected component separately, therefore we
  assume that $G$ is connected.


  Each vertex is placed on its bar.
  Figs.~\ref{fig:middleBar2Vertex} and~\ref{fig:bottomBar2Vertex}
  indicate how to route the adjacent half-edges.  Recall that the
  S-shape bend pairs can be eliminated.  Thus, a horizontal half-edge
  gets at most one extra bend and a vertical half-edge gets at most
  two extra bends; see Fig.~\ref{fig:bottomBar2Vertex}. We call a half-edge \emph{extreme} if it was
  horizontal and got one bend or vertical and got two bends that create a U-shape.

  \begin{figure}[tb]
    \centering
    \begin{subfigure}[b]{.31\linewidth}
      \centering
      \includegraphics[page=4]{middleBar2Vertex}
      \caption{}
      \label{fig:middleBar2Vertex_1}
    \end{subfigure}
    \hfill
    \begin{subfigure}[b]{.31\linewidth}
      \centering
      \includegraphics[page=2]{middleBar2Vertex}
      \caption{}
      \label{fig:middleBar2Vertex_2}
    \end{subfigure}
    \hfill
    \begin{subfigure}[b]{.31\linewidth}
      \centering
      \includegraphics[page=1]{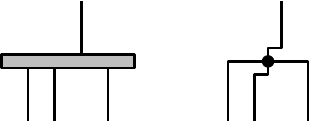}
      \caption{}
      \label{fig:middleBar2Vertex_3}
    \end{subfigure}

    \medskip

    \begin{subfigure}[b]{.31\linewidth}
      \centering
      \includegraphics[page=5]{middleBar2Vertex}
      \caption{}
      \label{fig:middleBar2Vertex_4}
    \end{subfigure}
    \hfill
    \begin{subfigure}[b]{.31\linewidth}
      \centering
      \includegraphics[page=6]{middleBar2Vertex}
      \caption{}
      \label{fig:middleBar2Vertex_5}
    \end{subfigure}
    \hfill
    \begin{subfigure}[b]{.31\linewidth}
      \centering
      \includegraphics[page=7]{middleBar2Vertex}
      \caption{}
      \label{fig:middleBar2Vertex_6}
    \end{subfigure}
    \caption{Replacing a middle bar with a vertex in the presence of
      (a)--(c)~zero, (d)--(e)~one, and (f)~two horizontal half-edges}
    \label{fig:middleBar2Vertex}
  \end{figure}

  \begin{figure}[tb]
    \centering
    \begin{subfigure}[b]{.25\linewidth}
      \centering
      \includegraphics[page=1]{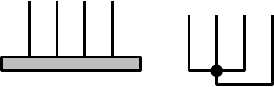}
      \caption{}
      \label{fig:bottomBar2Vertex_1}
    \end{subfigure}
    \hfill
    \begin{subfigure}[b]{.25\linewidth}
      \centering
      \includegraphics[page=2]{bottomBar2Vertex}
      \caption{}
      \label{fig:bottomBar2Vertex_2}
    \end{subfigure}
    \hfill
    \begin{subfigure}[b]{.3\linewidth}
      \centering
      \includegraphics[page=3]{bottomBar2Vertex}
      \caption{}
      \label{fig:bottomBar2Vertex_3}
    \end{subfigure}

    \medskip

    \begin{subfigure}[b]{.32\linewidth}
      \centering
      \includegraphics[page=4]{bottomBar2Vertex}
      \caption{}
      \label{fig:bottomBar2Vertex_4}
    \end{subfigure}
    \hfil
    \begin{subfigure}[b]{.32\linewidth}
      \centering
      \includegraphics[page=5]{bottomBar2Vertex}
      \caption{}
      \label{fig:bottomBar2Vertex_5}
    \end{subfigure}
    \caption{Replacing a bottom bar of degree $4$ with a vertex.}
    \label{fig:bottomBar2Vertex}
  \end{figure}
It suffices to show that the edges can be routed such that no edge
  is composed of two extreme half-edges. Even for red edges where we have the construction bend, we either get one extra bend from the horizontal (extreme) half-edge or two extra bends from the vertical (extreme) half-edge. Observe that an edge is
  extreme if and only if it is the rightmost or leftmost edge of a
  bottom or top bar, respectively, and it is attached to the bottom or
  top of the vertex, respectively. For each bottom or top
  bar we have the free choice to set either its rightmost or leftmost
  half-edge to become extreme.  Consider the following bipartite graph
  $H$. The vertices of $H$ are the top and bottom bars, as well as
  their leftmost and rightmost edges. A bar-vertex and an edge-vertex
  are adjacent in $H$ if and only if the bar and the edge are
  incident. Observe that each bar-vertex has degree two and each
  edge-vertex has degree at most two, thus $H$ is a union of disjoint
  paths and cycles and there is a matching of $H$ in which each
  bar-vertex is matched. 
This matching defines the extreme half-edges. It assigns exactly one half-edge to every bottom or top-bar and matches at most one half-edge of each edge. \end{proof}

\subsection{Orthogonal Drawings of Outer-1-Plane Graphs}
\label{sec:ortho.outer}

Since outer-1-planar graphs are planar
graphs~\cite{abbghnr-01pg-Alg16}, a planar orthogonal layout could be
computed with curve complexity at most three. For example, in
Fig.~\ref{fig:ortho_2_outer_counter_a} we can see an outer-1-plane
graph with a planar embedding in
Fig.~\ref{fig:ortho_2_outer_counter_b}.
Arguing similarly as we did for the proof of
Theorem~\ref{thm:biconnected1plane-not-OC3} it follows that there will be
at least two bends on an edge of the outer face. In this particular
case, Fig.~\ref{fig:ortho_2_outer_counter_c} shows an outer-1-planar
drawing of the same graph with at most two bends per edge. In the
following we compute 1-planar orthogonal layouts for biconnected
outer-1-planar graphs with optimal curve complexity three that also
preserve the initial outer-1-planar embedding.

\begin{theorem}
  \label{thm:outer1plane-not-OC2}
  Not every biconnected outer-$1$-plane graph of maximum degree~$4$
  admits an \OC{2}-layout.
\end{theorem}

\begin{figure}[tb]
  \centering
  \begin{subfigure}[b]{.35\linewidth}
    \centering {\includegraphics[page=6]{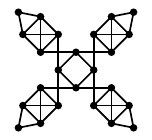}}
    \caption{outer-1-planar embedding}
    \label{fig:ortho_2_outer_counter_a}
  \end{subfigure}
  \hfill
  \begin{subfigure}[b]{.3\linewidth}
    \centering
    \includegraphics[page=7]{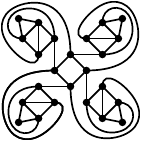}
    \caption{planar embedding}
    \label{fig:ortho_2_outer_counter_b}
  \end{subfigure}
  \hfill
  \begin{subfigure}[b]{.3\linewidth}
    \centering
    \includegraphics[page=8]{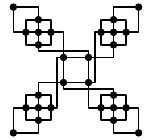}
    \caption{\OC{3}-layout}
    \label{fig:ortho_2_outer_counter_c}
  \end{subfigure}
  \caption{An outer-1-plane graph.}
  \label{fig:ortho_2_outer_example}
\end{figure}

\begin{figure}[tb]
  \centering
  \begin{subfigure}[b]{.3\linewidth}
    \centering
    \includegraphics[page=1]{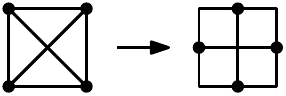}
    \caption{$K_4$}
    \label{FIG:OC2_outer_counter_eg_component}
  \end{subfigure}
  \hfil
  \begin{subfigure}[b]{.67\linewidth}
    \centering
    \includegraphics[page=5]{OC2_outer_counter_eg}
    \caption{two biconnected copies of $K_4$}
    \label{FIG:OC2_outer_counter_eg_composing}
  \end{subfigure}

  \caption{Biconnected outer-1-plane graph that does not admit an
    \OC{2}-layout with the same embedding.}
  \label{fig:ortho_2_outer_counter}
\end{figure}

\begin{proof}
  $K_4$ is a biconnected outer-1-plane graph.  Actually, it has a unique \OC{2}-layout as shown in
  Fig.~\ref{FIG:OC2_outer_counter_eg_component}.  When connecting two
  copies of $K_4$ by two intersecting edges as in
  Fig.~\ref{FIG:OC2_outer_counter_eg_composing}, it is not possible to
  draw the resulting graph such that the connector edges intersect and
  have curve complexity~two.
\end{proof}

\begin{restatable}{theorem}{orthoOuter}
  \label{thm:outer1plane-OC3}
  Every biconnected outer-$1$-plane graph of maximum degree~$4$ admits
  an \OC{3}-layout in an $ O(n) \times O(n)$ grid, where $n$ is the
  number of vertices.
\end{restatable}

\begin{proof}[sketch]
  Let $G$ be an outer-1-planar graph of maximum degree~4. Observe
  that all crossings  can be caged without
  changing the embedding: A maximal outer-1-planar graph always admits
  a straight-line outer-1-planar drawing in which all faces are
  convex~\cite{eggleton:86,dehkordi/eades:12}. We would directly obtain the
  required curve complexity if there were no top or bottom bars of
  degree~4. Instead, our proof is
  based on a 1-planar bar visibility representation of $G$ produced by
  a specific $st$-ordering. Let $s$ and $t$ be two vertices on the
  outer face. Define $S_l$ and $S_r$ to be the sequences of vertices on the
  left path and on the right path from $s$ to $t$ along the outer face
  of $G$, respectively. We choose $s,S_l,S_r,t$ as our $st$-ordering.
  Observe that this is also an st-ordering of the caged and planarized graph $G_p$.
  
  We process middle bars as in the algorithm of
  Theorem~\ref{thm:1-planar_OC4}. For the top and bottom bars of degree 4
  we choose differently which half-edge will be attached to
  the north or south port, respectively.  Let $v$ be a vertex such
  that $b(v)$ is a top or bottom bar of degree 4. Let $e_l=(v,v_l)$
  and $e_r=(v,v_r)$ be its leftmost and rightmost edge,
  respectively. Assume that $v \in S_l \cup \{s\}$ and  $b(v)$ is
  a bottom bar. If $v_l\in S_l$, we choose edge $e_l$ to be attached
  to the south port of $v$, otherwise we choose edge $e_r$. If $b(v)$
  is a top bar of degree 4 we choose its leftmost edge $e_l$ to be
  attached to the north port of $v$. Symmetrically, if $v\in S_r \cup
  \{t\}$ and $b(v)$ is a top bar, we choose $e_r$ for the north port
  of $v$ if $v_r\in S_r$, otherwise we choose $e_l$. If $b(v)$ is
  a bottom bar we choose its rightmost edge $e_r$  for the south port
  of $v$.
  
  The above choice has the following property (detailed proof in  \arxapp{Appendix~\ref{apx:thm5}}{\cite{arxivVersion}}): Any edge with three or four bends contains
  two consecutive bends that create an S-shape. The two bends are
  always connected with a vertical segment.  If this is an uncrossed
  edge of $G$, the S-shape can be eliminated.  For crossing edges, we
  prove that only one edge per crossing may have more than two
  bends. If the vertical segment connecting the two bends of the
  S-shape is crossed, we apply the flow technique of
  Tamassia~\cite{tamassia87} around the crossing point and reduce the
  number of bends (for details refer to \arxapp{Appendix~\ref{apx:thm5}}{\cite{arxivVersion}}).
\end{proof}

\section{Smooth Orthogonal 1-Planar Drawings}
\label{sec:smooth}

In this section we examine smooth orthogonal 1-planar
drawings. In particular, we show that every\improvement{ biconnected}{} 1-plane graph
of maximum degree~4 admits an \improvement{\SC4}{\SC3}-layout that preserves the given
embedding\improvement{, and an \SC3-layout if we are allowed to change the
embedding}{}.  For biconnected outer-1-plane graphs, we achieve \SC2,
which is optimal for this graph class.

\subsection{Smooth Orthogonal Drawings for General 1-Planar Graphs}
\improvement{
\begin{theorem}
  \label{thm:biconnected1plane-SC4}
  Every biconnected $1$-plane graph of maximum degree~$4$ has an
  \SC{4}-layout.  The drawing area may be super-polynomial.
\end{theorem}
\begin{proof}
  After planarizing the given graph, we can apply the algorithm of
  Alam et al.~\cite{ABKKKW14} on the planar graph in order to produce
  an \SC{2}-layout
  respecting the given planar embedding.  In the resulting drawing,
  any original edge that is involved in a crossing may obtain up to
  two segments from each of its half-edges.  The drawing area is the
  same as that of the algorithm of Alam et al.~\cite{ABKKKW14}.
\end{proof}
}{}
\begin{figure}[tb]
  \centering
  \begin{subfigure}[b]{.15\linewidth}
    \centering
    \includegraphics[page=2]{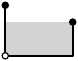}
    \caption{}
    \label{fig:nonIntersectedUShapesDummy}
  \end{subfigure}
  \hfil
  \begin{subfigure}[b]{.15\linewidth}
    \centering
    \includegraphics[page=5]{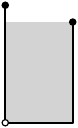}
    \caption{}
    \label{fig:nonIntersectedUShapesEmptyRegion}
  \end{subfigure}
  \hfil
  \begin{subfigure}[b]{.15\linewidth}
    \centering
    \includegraphics[page=7]{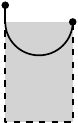}
    \caption{}
    \label{fig:nonIntersectedUShapesSmooth}
  \end{subfigure}

  \caption{Smoothing process of U-shapes created by top (bottom)
    bars.}
  \label{fig:intersectedUShapes}
\end{figure}

\begin{theorem}
  \label{thm:biconnected1planar-SC3}
  Every \improvement{biconnected $1$-planar}{$1$-plane} graph of maximum degree~$4$ admits an
  \SC{3}-layout in $O(n) \times O(n^2)$ area.
\end{theorem}

\begin{proof}
  We compute an \SC{3}-layout based on an
  \OC{4}-layout computed by the algorithm of
  Theorem~\ref{thm:1-planar_OC4}. Observe that in the \OC{4}-layouts
  calculated by our approach, the area bounded U-shaped half-edges created at top and
  bottom bars is vertex-free (see gray area in Fig.~\ref{fig:nonIntersectedUShapesDummy}), and, each vertex is located on a separate
  level. We replace one bend of each U-shaped half-edge by a dummy
  vertex; see Fig.~\ref{fig:nonIntersectedUShapesDummy}.
  By doing so, we split each U-shaped half-edge into a vertical edge
  and an $L$-shaped half-edge. In the following, we treat the
  $L$-shaped half-edge as if the bend was on an $L$-shaped half-edge
  incident to the dummy vertex. We process $V=\{v_1,v_2,\ldots,v_n\}$
  in the ascending vertical order of vertices (including dummy
  vertices). For $v_i$, let $\Delta^\uparrow_i$ be the largest
  horizontal distance between $v_i$ and any bend on incident
  $L$-shaped half-edges leading to neighbors with larger
  index. Let $\Delta^\downarrow_i$ be the corresponding value
  for bends at incident $L$-shaped half-edges and construction bends
  of red edges incident to edges leading to neighbors with smaller
  index. We increase the y-coordinate of all $v_j$ with $j \geq i$ by
  $\Delta^\downarrow_i$ units and then the y-coordinate of all $v_k$
  with $k > i$ by $\Delta^\uparrow_i$ units. Bends on $L$-shaped
  half-edges and construction bends of
  red edges leading to neighbors with smaller index will be moved
  together with the corresponding vertex. Note that the region enclosed 
  by U-shapes created at top
  and bottom bars remains empty; see
  Fig.~\ref{fig:nonIntersectedUShapesEmptyRegion}. After the
  stretching, we remove the additional dummy vertices.

  Each U-shaped half-edge will be replaced by a semi-circle which
  fits into the corresponding stretched empty region. We place the semi-circle
  directly incident to the endpoint which created the U-shape; see
  Fig.~\ref{fig:nonIntersectedUShapesSmooth}. Then we replace each
  intersected $S$-shaped half-edge formed by a construction bend
   of a red edge by two
  consecutive quarter arcs incident to the top endpoint of the
  edge. Recall that if a red edge has an $S$-shape from its top
  vertex, it has no bend from its bottom vertex.  Further we replace
  each remaining bend by a quarter arc starting at the corresponding
  endpoint. Arcs at the two endpoints will be connected by a vertical
  segment. The correctness follows from the fact that the regions
  stretched to make space for drawing arcs were empty in the initial
  drawing.

  The area of the resulting drawing is $O(n) \times O(n^2)$ as the
  input drawing had $O(n)\times O(n)$ area and for every vertex the stretching operation increases the height by at most the length of the longest horizontal segment (i.e. $O(n)$).
\end{proof}

\subsection{Smooth Orthogonal Drawings for Outer-1-Plane Graphs}
\label{sec:smooth.outer}

We focus on smooth layouts of outer-1-plane graphs. We demonstrate
that curve complexity one is not always possible, but curve complexity
two can be achieved for biconnected outer-1-plane graphs.  We start
with the following observation.  The complete graph on four vertices
with free ports towards its outer face has a unique \SC{1}-layout,
shown in Fig.~\ref{fig:K4-sc1}. Removing one edge and restricting all
ports towards its outer face, there exist two \SC{1}-layouts, see
Figs.~\ref{fig:K4-e-sc1a} and~\ref{fig:K4-e-sc1b}.

\begin{figure}[tb]
  \centering
  \begin{subfigure}[b]{.12\linewidth}
    \centering
    \includegraphics[page=4]{K4-e}
    \caption{}
    \label{fig:K4-sc1}
  \end{subfigure}
  \hfill
  \begin{subfigure}[b]{.12\linewidth}
    \centering
    \includegraphics[page=2]{K4-e}
    \caption{}
    \label{fig:K4-e-sc1a}
  \end{subfigure}
  \hfill
  \begin{subfigure}[b]{.1\linewidth}
    \centering
    \includegraphics[page=3]{K4-e}
    \caption{}
    \label{fig:K4-e-sc1b}
  \end{subfigure}
  \hfill
  \begin{subfigure}[b]{.15\linewidth}
    \centering
    \includegraphics{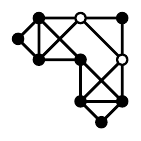}
    \caption{}
    \label{fig:no-sc1}
  \end{subfigure}
  \hfill
  \begin{subfigure}[b]{.13\linewidth}
    \centering
    \includegraphics[page=2]{K4-e-combined}
    \caption{}
    \label{fig:no-sc1-em1}
  \end{subfigure}
  \hfill
  \begin{subfigure}[b]{.13\linewidth}
    \centering
    \includegraphics[page=3]{K4-e-combined}
    \caption{}
    \label{fig:no-sc1-em2}
  \end{subfigure}
  \hfill
  \begin{subfigure}[b]{.13\linewidth}
    \centering
    \includegraphics[page=4]{K4-e-combined}
    \caption{}
    \label{fig:no-sc1-em3}
  \end{subfigure}
  \caption{(a)~\SC{1}-layouts for $K_4$ and (b)--(c)~for $K_4-e$ with
    restricted ports.  (d)~A biconnected outer-1-plane graph that does
    not have an \SC{1}-layout. (e)-(g) \SC{1}-layouts of a subgraph of (d).}
\end{figure}

\begin{theorem}
  \label{thm:biconnectedouter1plane-not-SC1}
  Not every biconnected outer-$1$-plane graph of maximum degree~$4$ has an
  \SC{1}-layout.
\end{theorem}

\begin{proof}
  Take the graph in Fig.~\ref{fig:no-sc1}.  It has two subgraphs
  isomorphic to $K_4-e$ (with restricted ports) that share a vertex. Combining two drawings for both copies gives rise to the three drawings in Figs.~\ref{fig:no-sc1-em1}--\ref{fig:no-sc1-em3} in which the edge between the two highlighted vertices cannot be added with curve complexity one.
\end{proof}

To achieve \SC2-layouts for biconnected outer-1-plane
graphs\improvement{,}{ (see Fig.~\ref{FIG:sc2} for an example),} 
we modify the algorithm of Alam~et~al.\ \cite{ABKKKW14} for
outerplane graphs; see \arxapp{Appendix~\ref{apx:thm9}}{\cite{arxivVersion}} for details.

\begin{restatable}{theorem}{smoothOuter}\label{thm:biconnectedouter1plane-SC2}
  Every biconnected outer-$1$-plane graph of maximum degree~$4$ has an
  \SC{2}-layout.  The drawing area may be super-polynomial.
\end{restatable}

\begin{proof}[sketch]
  The algorithm of Alam et al.~\cite{ABKKKW14} processes the faces of
  the graph along the \emph{weak-dual}, i.e., the dual graph omitting
  the outer face and rooted at some inner face.  For the next face,
  one of its edges (the \emph{reference edge}) is already
  drawn and imposes the drawing of the face.
  Figures~\ref{fig:alamEtAlOuterSmooth1ax}--\ref{fig:alamEtAlOuterSmooth3bx}
  show the different cases.

  We define an auxiliary graph $G'$: Let $G$ be a
  biconnected outer-1-plane graph, and let~$G_\mathrm{p}$ be the planarized
  graph of $G$, where crossing points are replaced with dummy
  vertices. Three types of dummy vertices exist in $G_\mathrm{p}$:
  \emph{dummy-cuts} (cut vertices), \emph{in-dummies} (only incident
  to inner faces), and \emph{out-dummies}.
  $G'$ contains all in-dummy and out-dummy
  vertices of $G_\mathrm{p}$, while dummy-cuts are replaced by a caging
  cycle. The face inside a caging cycle is called a \emph{cut-face}.
  All other faces are called \emph{normal}. Faces are processed
  along a traversal of the weak dual of $G'$. As $G'$ may not be
  outerplanar, its weak dual does not have to be
  acyclic. It contains cycles of length four around
  in-dummies (see Fig.~\ref{fig:virtualEdgesInDummyx}).  The auxiliary
  graph~$G'$ also contains \emph{virtual edges} that are red. These
  are edges added for caging dummy-cuts and edges added to complete
  the process of faces around an in-dummy.
  Figures~\ref{fig:normal1x}--\ref{fig:normal5x} show how to process
  normal faces not appearing in Alam et
  al.~\cite{ABKKKW14}. When processing a
  cut-face, we draw the crossing edges instead of the caging
  cycles; see
  Figs.~\ref{fig:smooth_cut1x}--\ref{fig:smooth_cut_virtual5x} for two
  out of ten cases. Finally, in order to draw the fourth face around
  an in-dummy, we ensure that the edge-segments incident
  to the dummy vertex have the same length; see
  Fig.~\ref{fig:facial_pair3x} for an example.
\end{proof}

\begin{figure}[tb]
  \centering
  \begin{subfigure}[b]{.135\linewidth} 
    \centering
    \includegraphics[page=1,width=\textwidth]{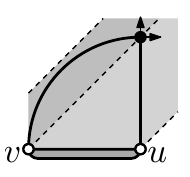}
    \caption{}
    \label{fig:alamEtAlOuterSmooth1ax}
  \end{subfigure}
  \hfill
  \begin{subfigure}[b]{.135\linewidth} 
    \centering
    \includegraphics[page=3,width=\textwidth]{alam-outer-smooth}
    \caption{}
    \label{fig:alamEtAlOuterSmooth2ax}
  \end{subfigure}
  \hfill
  \begin{subfigure}[b]{.135\linewidth} 
    \centering
    \includegraphics[page=5,width=\textwidth]{alam-outer-smooth}
    \caption{}
    \label{fig:alamEtAlOuterSmooth3ax}
  \end{subfigure}
  \hfill
  \begin{subfigure}[b]{.135\linewidth} 
    \centering
    \includegraphics[page=2,width=\textwidth]{alam-outer-smooth}
    \caption{}
    \label{fig:alamEtAlOuterSmooth1bx}
  \end{subfigure}
  \hfill
  \begin{subfigure}[b]{.135\linewidth} 
    \centering
    \includegraphics[page=4,width=\textwidth]{alam-outer-smooth}
    \caption{}
    \label{fig:alamEtAlOuterSmooth2bx}
  \end{subfigure}
  \hfill
  \begin{subfigure}[b]{.135\linewidth} 
    \centering
    \includegraphics[page=6,width=\textwidth]{alam-outer-smooth}
    \caption{}
    \label{fig:alamEtAlOuterSmooth3bx}
  \end{subfigure}

  \smallskip

  \begin{subfigure}[b]{.20\linewidth} 
    \centering
    \includegraphics[page=1, width=\textwidth]{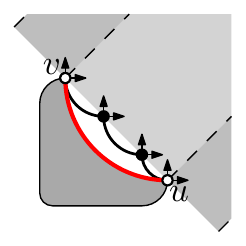}
    \caption{}
    \label{fig:normal1x}
  \end{subfigure}
  \qquad
  \begin{subfigure}[b]{.20\linewidth} 
    \centering
    \includegraphics[page=2, width=\textwidth]{smooth_normal}
    \caption{}
    \label{fig:normal2x}
  \end{subfigure}
  \qquad
  \begin{subfigure}[b]{.20\linewidth} 
    \centering
    \includegraphics[page=3, width=\textwidth]{smooth_normal}
    \caption{}
    \label{fig:normal3x}
  \end{subfigure}
  \qquad
  \begin{subfigure}[b]{.20\linewidth} 
    \centering
    \includegraphics[page=5, width=\textwidth]{smooth_normal}
    \caption{}
    \label{fig:normal5x}
  \end{subfigure}

  \smallskip

  \begin{subfigure}[b]{.22\linewidth} 
    \centering
    \includegraphics[page=1]{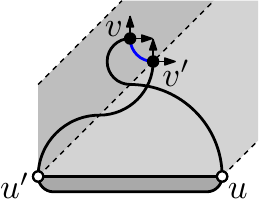}
    \caption{}
    \label{fig:smooth_cut1x}
  \end{subfigure}
  \hfill
  \begin{subfigure}[b]{.23\linewidth} 
    \centering
    \includegraphics[page=5]{smooth_cut_virtual}
    \caption{}
    \label{fig:smooth_cut_virtual5x}
  \end{subfigure}
  \hfill
  \begin{subfigure}[b]{.23\linewidth} 
    \hfill
    \includegraphics[page=2]{outer-smooth}
    \caption{}
    \label{fig:virtualEdgesInDummyx}
  \end{subfigure}
  \hfill
  \begin{subfigure}[b]{.18\linewidth} 
    \centering
    \includegraphics[page=4]{facial_pair}
    \caption{}
    \label{fig:facial_pair3x}
  \end{subfigure}
  \caption{Constructing an \SC{2}-drawing of biconnected outer
    1-planar graphs.}
  \label{fig:alamEtAlOuterSmoothx}
\end{figure}

\improvement{}{
\begin{figure}
  \begin{center}
    \includegraphics[trim=0 170 0 0, clip, scale=1]{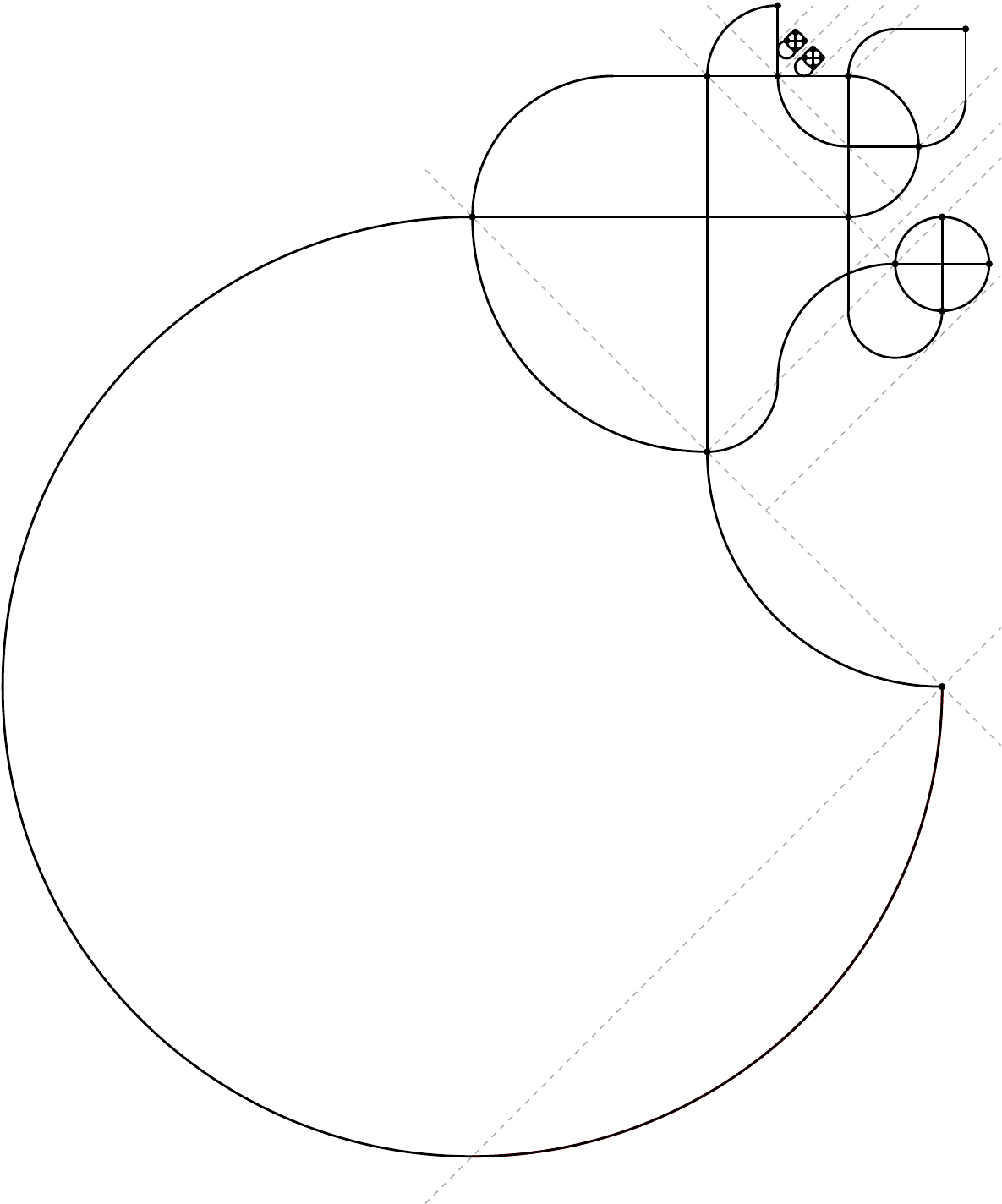}
  \end{center}
  \caption{\label{FIG:sc2}\SC2-layout of an outer-1-plane graph. Largest 3/4-arc only partially drawn.}
\end{figure}
}

\section{A List of Open Problems}
\label{sec:future}
\begin{itemize}
\item Can we improve our curve complexity bounds if we restrict
  ourselves to more strongly connected classes of graphs (of maximum
  degree~4)?
\item Candidate subclasses of outer-1-plane graphs for \SC1-layouts
  are for example outer-IC-plane graphs where crossings are
  independent.  A possible variant would be to allow degenerate
  layouts where pairs of edges can touch but not cross.
\item \improvement{Do $1$-plane graphs always admit an \OC4 layout?
  Are there lower bounds on the curve complexity for orthogonal
  layouts of $1$-planar graphs and smooth orthogonal layouts of
  $1$-plane and $1$-planar graphs.}{Is there a 1-plane graph that does
  not admit an \SC2-layout?
\item Do biconnected outer-1-plane graphs admit an \SC2-layout with
  polynomial drawing area?}
\item Do similar results also hold for $2$-planar
  graphs and more generally beyond-planar graphs?
\end{itemize}


\clearpage

\bibliographystyle{splncs04} 
\bibliography{abbrv,smooth-v02}


\arxapp{\newpage
\appendix

\section*{Appendix}

\section{Additional Material for Section~\ref{sec:ortho.outer}}
\label{apx:thm5}

\orthoOuter*

\begin{proof}
  Let $G$ be an outer-1-planar graph of maximum degree~4.  We want to
  use again a 1-planar bar visibility representation.  First observe
  that all crossings in an outer-1-planar graph can be caged without
  changing the embedding: A maximal outer-1-planar graph always admits
  a straight-line outer-1-planar drawing in which all faces are
  convex~\cite{eggleton:86,dehkordi/eades:12}. We would obtain the
  required curve complexity if there were no top or bottom bars of
  degree~4. Instead we will work with a specialized $st$-ordering.

  Let $s$ and $t$ be two vertices on the outer face. Let $S_l$ and
  $S_r$ be the vertices on the left path and the right path from $s$
  to $t$ along the outer face of $G$, respectively. Note that due to
  biconnectivity each vertex appears at most once in $S_l$ or $S_r$.
  We choose $s,S_l,S_r,t$ as the $st$-ordering for the construction of
  the 1-planar bar-visibility representation of $G$.

  We want to replace every bar in a similar way as in the proof of
  Theorem~\ref{thm:1-planar_OC4}. For the top and bottom bars of
  degree 4 we make a different choice based on which half-edge will be
  attached to the north or south port, respectively.

  Let $v$ be a vertex such that $b(v)$ is a top or bottom bar of
  degree 4. Let $e_l=(v,v_l)$ and $e_r=(v,v_r)$ be its leftmost and
  rightmost edges, respectively. Assume first $v \in S_l \cup \{s\}$
  and that $b(v)$ is a bottom bar. If $v_l\in S_l$, we select edge
  $e_l$ to be attached to the south port of $v$, otherwise we select
  edge $e_r$. If $b(v)$ is a top bar of degree 4, then all its
  neighbors appear before $v$ in the $st$-ordering and they belong to
  $S_l$, hence we choose its leftmost edge $e_l$ to be attached to the
  north port of $v$. Symmetrically, if $v\in S_r \cup \{t\}$ and
  $b(v)$ is a top bar, we choose $e_r$ for the north port of $v$ if
  $v_r\in S_r$, otherwise we choose $e_l$. And if $b(v)$ is a bottom
  bar, all its neighbors appear after $v$ and we choose its rightmost
  edge $e_r$ for the south port of $v$.

  We claim that the above choice creates at most four bends, and in
  the case where three or four bends appear, two of them create an
  S-shape and are vertically aligned. Note that if the vertical
  segment connecting the two bends is not crossed, then the two bends
  can be eliminated and the theorem holds.  So, consider an edge
  $e=(u,v)$ such that $u$ has a lower index than $v$. Three or more
  bends appear only if edge $e$ uses the south port of $u$ and/or the
  north port of $v$. There are three cases, depending on whether $u$,
  $v$ belong to $S_l \cup \{s\}$ or $S_r \cup \{t\}$.

  {\bf Case 1:} Suppose that $u,v\in S_l \cup \{s\}$. Assume first
  that $e$ uses the south port of $u$. Then $b(u)$ is a degree~4
  bottom bar and $e$ is the leftmost edge of $b(u)$. All other
  neighbors of $u$ come after $v$ in the $st$-ordering. If there
  exists another edge attached to the bottom of $b(v)$ then this edge
  can only be incident to vertices with indices between $u$ and $v$,
  and therefore $e$ is the rightmost edge at the bottom of $b(v)$;
  refer to Fig.~\ref{fig:outer-1-planeOrtho1}. When replacing $b(v)$
  with vertex $v$, edge $e$ can only use the south or east port of
  $v$. This is true even in the case where $b(v)$ is a degree~4 top
  bar since in that case the north port will be used by the leftmost
  edge of $b(v)$ and not by $e$. There are three bends only if $e$
  uses the east port of $v$ (see
  Fig.~\ref{fig:outer-1-planeOrtho2}). Two of them form an S-shape
  and are connected by a vertical segment as claimed.

  Assume now that $e$ uses the north port of $v$. Then $b(v)$ is a
  degree~4 top bar and $e$ is the leftmost edge of $b(v)$. All other
  neighbors of $v$ come before $u$ in the $st$-ordering, and arguing
  similarly as before, we can conclude that $e$ cannot be the leftmost
  edge at the top of $b(u)$. Hence, $e$ will use either the north or the
  east port of $u$. Three bends are created only if $e$ uses the east
  port of $u$ (see Fig.~\ref{fig:outer-1-planeOrtho3}) and the claim
  holds.
	
  {\bf Case 2:} The case where $u,v\in S_r \cup \{t\}$ is similar to
  the case where $u,v\in S_l \cup \{s\}$ and is depicted in
  Figs.~\ref{fig:outer-1-planeOrtho4}-\ref{fig:outer-1-planeOrtho6}.
	
  {\bf Case 3:} The last case that remains to consider for our claim,
  is the case where $u\in S_l \cup \{s\}$ and $v\in S_r \cup \{t\}$.
  Here, if $e$ uses the south port of $u$, then $b(u)$ is a degree~4
  bottom bar and $e$ is its rightmost edge.  It is not hard to see
  that $e$ is
  either the left edge of $b(v)$ or
  the leftmost edge attached to the bottom of $b(v)$. Therefore $e$
  will not use the east port of $v$. Similarly, if $e$ uses the north
  port of $v$ then $e$ is the leftmost edge of the degree~4 top bar
  $b(v)$, and $e$ is the rightmost edge of the top of $b(u)$ and
  cannot use the west port of $u$. In any case, $e$ has at most four
  bends and satisfies the claim; refer to
  Figs.~\ref{fig:outer-1-planeOrtho7}-\ref{fig:outer-1-planeOrtho8}
  for the case where $e$ is drawn with four bends.

  \begin{figure}[tb]
    \centering
    \begin{subfigure}[b]{.13\linewidth}
      \centering
      \includegraphics[page=1, width=\textwidth]{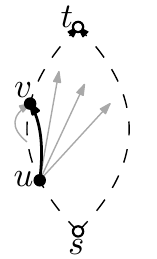}
      \caption{}
      \label{fig:outer-1-planeOrtho1}
    \end{subfigure}
    \hfil
    \begin{subfigure}[b]{.13\linewidth}
      \centering
      \includegraphics[page=2, width=\textwidth]{outer-ortho_new}
      \caption{}
      \label{fig:outer-1-planeOrtho2}
    \end{subfigure}
    \hfil
    \begin{subfigure}[b]{.13\linewidth}
      \centering
      \includegraphics[page=3, width=\textwidth]{outer-ortho_new}
      \caption{}
      \label{fig:outer-1-planeOrtho3}
    \end{subfigure}
    \hfil
    \begin{subfigure}[b]{.13\linewidth}
      \centering
      \includegraphics[page=4, width=\textwidth]{outer-ortho_new}
      \caption{}
      \label{fig:outer-1-planeOrtho4}
    \end{subfigure}
    \hfil
    \begin{subfigure}[b]{.13\linewidth}
      \centering
      \includegraphics[page=5, width=\textwidth]{outer-ortho_new}
      \caption{}
      \label{fig:outer-1-planeOrtho5}
    \end{subfigure}
    \hfil
    \begin{subfigure}[b]{.13\linewidth}
      \centering
      \includegraphics[page=6, width=\textwidth]{outer-ortho_new}
      \caption{}
      \label{fig:outer-1-planeOrtho6}
    \end{subfigure}
    \hfil
    \begin{subfigure}[b]{.13\linewidth}
      \centering
      \includegraphics[page=7, width=\textwidth]{outer-ortho_new}
      \caption{}
      \label{fig:outer-1-planeOrtho7}
    \end{subfigure}
    \hfil
    \begin{subfigure}[b]{.13\linewidth}
      \centering
      \includegraphics[page=8, width=\textwidth]{outer-ortho_new}
      \caption{}
      \label{fig:outer-1-planeOrtho8}
    \end{subfigure}
    \hfil
    \begin{subfigure}[b]{.13\linewidth}
      \centering
      \includegraphics[page=11, width=\textwidth]{outer-ortho_new}
      \caption{}
      \label{fig:outer-1-planeOrtho11}
    \end{subfigure}
    \hfil
    \begin{subfigure}[b]{.13\linewidth}
      \centering
      \includegraphics[page=9, width=\textwidth]{outer-ortho_new}
      \caption{}
      \label{fig:outer-1-planeOrtho9}
    \end{subfigure}
    \hfil
    \begin{subfigure}[b]{.13\linewidth}
      \centering
      \includegraphics[page=10, width=\textwidth]{outer-ortho_new}
      \caption{}
      \label{fig:outer-1-planeOrtho10}
    \end{subfigure}
    \hfil
    \begin{subfigure}[b]{.13\linewidth}
      \centering
      \includegraphics[page=1]{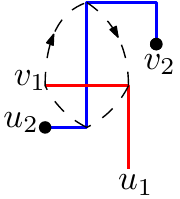}
      \caption{}
      \label{fig:outer-1-planeOrtho9a}
    \end{subfigure}
    \hfil
    \begin{subfigure}[b]{.13\linewidth}
      \centering
      \includegraphics[page=2]{wings}
      \caption{}
      \label{fig:outer-1-planeOrtho10b}
    \end{subfigure}
    \caption{Cases considered in the proof of
      Theorem~\ref{thm:outer1plane-OC3}.}
    \label{fig:outer-1-planeBAndWConfiguration}
  \end{figure}

  As already mentioned the above claim implies that if the vertical
  segment connecting the two bends of the S-shape is not crossed, then
  the S-shape can be eliminated so that the theorem holds.  This is true
  if edge $e=(u,v)$ is a planar edge of $G$ or if it is a crossing red
  edge.
  It remains to consider blue edges.
  Recall that the selection of red and blue edges for the construction
  of the 1-planar bar visibility representation assured that the
  vertically drawn blue edge is always incident to the topmost bar.
  Let $e_1=(u_1,v_1)$ and $e_2=(u_2,v_2)$ be two crossing edges, such
  that $u_1$ appears before all other vertices in the $st$-ordering, and
  $u_2$ appears before $v_2$. We distinguish three cases depending on
  whether we have a diamond configuration, a left wing, or a right wing
  configuration in the 1-planar bar visibility representation.
  \begin{itemize}
  \item $e_1$ and $e_2$ create a diamond configuration. In this case we
    have $u_1,u_2\in S_l\cup{\{s\}}$ and $v_1,v_2\in S_r\cup{\{t\}}$ as
    shown in Fig.~\ref{fig:outer-1-planeOrtho11}. Due to 1-planarity
    only edge $e_2$ can be drawn with three or four bends and this is
    always the red edge of a diamond configuration in the 1-planar bar
    visibility representation.
  \item $e_1$ and $e_2$ create a left wing configuration;
    refer to
    Fig.~\ref{fig:outer-1-planeOrtho9}.
    In this case $e_2$ is the blue edge and
     vertices $u_1,u_2,v_1$ are in $S_l\cup{\{s\}}$.  By
    outer-1-planarity, $b(u_2)$, cannot have a top edge to the right
    of $e_2$ nor an additional edge to $S_r \cup t$.  Thus, $e_2$
    cannot be attached to the west or the south port of $u_2$. Hence,
    if $e_2$ has more than two bends, then $b(v_2)$ must be a degree 4
    top bar, $e_2$ uses the north port of $v_2$ and the east port
    of~$u_2$.  
    Figure~\ref{fig:outer-1-planeOrtho9a} shows the drawing of the two
    crossing edges.  Now consider the red edge $e_1$. By
    outer-1-planarity,
    $b(u_1)$ cannot have a top edge to the right of $e_1$ nor an
    additional edge to $S_r \cup t$. Thus, $e_1$ cannot be attached to
    the west or the south port of $u_1$. Similarly, $e_1$ is not attached to the north or west port of $v_1$ and the construction bend of $e_1$ was not removed due to an S-shaped pair of
    bends. We apply the flow technique around the crossing point of $e_1$ and $e_2$ as indicated in Fig.~\ref{fig:outer-1-planeOrtho9a}: two bends of $e_2$ are removed and the construction bend of $e_1$ is moved to the other side of the crossing.
    
    \item The case where $e_1$ and $e_2$ create a right wing is symmetric to the previous case and indicated in Figs.~\ref{fig:outer-1-planeOrtho10} and~\ref{fig:outer-1-planeOrtho10b}. 
  \end{itemize}
  We showed that whenever an edge has three or four bends, then two
  bends create an S-shape and are connected with a vertical
  segment. The two bends can be removed from the drawing giving an
  orthogonal 1-plane drawing with curve complexity three, as the theorem
  states.
\end{proof}

\section{Additional Material for
  Section~\ref{sec:smooth.outer}}
\label{apx:thm9}

In order to achieve curve complexity two for smooth orthogonal
drawings of biconnected outer-1-plane graphs, we modify the
algorithm of Alam et al.~\cite{ABKKKW14} for outerplane graphs. Hence,
in the following, we show how their \SC{1}-layout algorithm for
outerplane graphs deals with biconnected outerplane graphs; for
details we refer to the original paper~\cite{ABKKKW14}.

In order to define an ordering of the faces, the algorithm of Alam et
al. uses the weak dual tree $T$ of $G$ which is rooted at a leaf
face. One edge of the root face incident to a degree two vertex and a
vertex of degree at most three\footnote{If the graph does not have a
  leaf face with an edge of this property, removing a degree~$2$
  vertex will produce a new leaf face with the required property.} is
selected as the first edge which will be drawn as a vertical
segment. Following $T$ the graph is drawn face by face. Note that for
each face that we draw, we have previously already drawn a single edge
$(u,v)$ which will serve as a reference to select a suitable case from
Fig.~\ref{fig:alamEtAlOuterSmooth}. Observe that in the cases shown in
Figs.~\ref{fig:alamEtAlOuterSmooth1a},
\ref{fig:alamEtAlOuterSmooth2a}, \ref{fig:alamEtAlOuterSmooth1b}
and~\ref{fig:alamEtAlOuterSmooth2b} a \emph{side-arc} is introduced,
that is, a convex quarter circle which will not serve as a new
reference edge since it is incident to $u$ or $v$ which has remaining
degree zero by construction.
\begin{figure}[tb]
  \centering
  \begin{subfigure}[b]{.135\linewidth}
    \centering
    \includegraphics[page=1, width=\textwidth]{alam-outer-smooth}
    \caption{}
    \label{fig:alamEtAlOuterSmooth1a}
  \end{subfigure}
  \hfill
  \begin{subfigure}[b]{.135\linewidth}
    \centering
    \includegraphics[page=3, width=\textwidth]{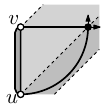}
    \caption{}
    \label{fig:alamEtAlOuterSmooth2a}
  \end{subfigure}
  \hfill
  \begin{subfigure}[b]{.135\linewidth}
    \centering
    \includegraphics[page=5, width=\textwidth]{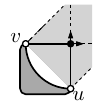}
    \caption{}
    \label{fig:alamEtAlOuterSmooth3a}
  \end{subfigure}
  \hfill
  \begin{subfigure}[b]{.135\linewidth}
    \centering
    \includegraphics[page=2, width=\textwidth]{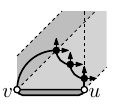}
    \caption{}
    \label{fig:alamEtAlOuterSmooth1b}
  \end{subfigure}
  \hfill
  \begin{subfigure}[b]{.135\linewidth}
    \centering
    \includegraphics[page=4, width=\textwidth]{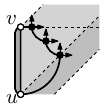}
    \caption{}
    \label{fig:alamEtAlOuterSmooth2b}
  \end{subfigure}
  \hfill
  \begin{subfigure}[b]{.135\linewidth}
    \centering
    \includegraphics[page=6, width=\textwidth]{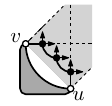}
    \caption{}
    \label{fig:alamEtAlOuterSmooth3b}
  \end{subfigure}
  \caption{Algorithm of Alam et al.~\cite{ABKKKW14} for \SC{1}-drawing
    outerplane graphs w.r.t.\ reference edge $(u,v)$: (a)--(c)
    inserting singletons and (d)--(f) inserting chains.}
  \label{fig:alamEtAlOuterSmooth}
\end{figure}

Planarity is proven as follows: When inserting a face above reference
edge $(u,v)$, additionally a semi-strip $L_{u,v}$ bounded by rays of
slope $+1$ emerging from $u$ and $v$ is introduced (lightgray in
Fig.~\ref{fig:alamEtAlOuterSmooth}). Note that the two rays are not
part of the lightgray semi-strip. If the face contains a \emph{side-arc} an
additional semi-strip touching the lightgray one is introduced
(dark-gray in Fig.~\ref{fig:alamEtAlOuterSmooth}). This semi-strip has
half the width of the lightgray one. The entire subgraph that can be
separated by removing $u$ and $v$ will be located inside $L_{u,v}$ and, potentially,
its two surrounding dark-gray semi-strips $L_{u,v}^t$ and
$L_{u,v}^b$. In particular, let $(u^\prime,v^\prime)$ be the reference edge for
the parent face in the weak dual $T$ that was used for placing $u$ and
$v$. Then, all semi-strips defined by $u$ and $v$ will be contained in
$L_{u',v'}$ and $L_{u',v'}^t$ (if it exists) and $L_{u',v'}^b$ (if it
exists). Note that due to the number of ports available, if the
subgraph incident to $u$ and $v$ uses $L_{u,v}^t$ or $L_{u,v}^b$, the
corresponding edge neighboring $(u,v)$ cannot introduce another
subgraph, hence no additional semi-strip is defined. Therefore, semi-strips can
overlap only if they are defined by two faces, $f_1$ and $f_2$, such that $f_1$
is ancestor of $f_2$ in $T$, which prevents intersections.

Consider now a biconnected outer-1-planar graph $G$ and the planar graph
$G_\mathrm{p}$ derived from $G$ by replacing all crossings with dummy
vertices. In $G_\mathrm{p}$, there exist three types of dummy vertices: a
\emph{dummy-cut}, which is a cut vertex of $G_\mathrm{p}$ and its four edges
belong to the outer face of $G_\mathrm{p}$, an \emph{out-dummy}, which is not a
cut vertex but is located on the outer face of $G_\mathrm{p}$ (has exactly two
consecutive edges on the outer face), and an \emph{in-dummy}, which is
not on the outer face of $G_\mathrm{p}$.

\smoothOuter*

\begin{proof}
  We use a modified version of the algorithm of Alam et
  al.~\cite{ABKKKW14} which produces \SC{1}-layouts for outerplanar
  graphs without crossings.  The algorithm requires an ordering of the
  vertices which is computed by the weak dual of a biconnected graph
  $G'$ and an appropriate starting edge as reference for the first
  face. We first introduce a suitable order of the vertices which
  assumes that the first edge is selected appropriately.

  Let $G$ be a biconnected outer-1-plane graph, and let~$G_\mathrm{p}$ the
  planarized graph of $G$. We define a biconnected graph $G'$ as
  follows. We keep all in-dummy and out-dummy vertices of $G_\mathrm{p}$ in
  $V(G')$. For a dummy-cut $x$ created by a pair of crossing edges
  $(u,v)$ and $(u',v')$, we create a 4-wheel around $x$ by adding
  \emph{virtual edges} as shown in Fig.~\ref{fig:virtualEdgesDummyCut}
  and remove $x$. Note that for each dummy-cut, we add at least two
  virtual edges that are on the outer face of $G'$.  Now $G'$ is
  planar (not necessarily outerplanar), and contains only in-dummy
  vertices in its interior. If a face was created by a dummy-cut, we
  say it is a \emph{cut-face}, otherwise it is a \emph{normal} face.
  Let $e$ be the starting reference edge on the outer face of $G'$
  that is also incident to another face $f_0$. The weak dual of $G'$
  may contain cycles that are created by in-dummy vertices. However
  any cycle has length four and any two cycles are edge disjoint due to
  1-planarity.  We order the faces of $G'$ by applying leftmost BFS on
  its weak dual and starting from $f_0$. Cycles of length four have two
  directed paths of the same length, say $f_1,f_2,f_4$ and
  $f_1,f_3,f_4$ where $f_2$ and $f_3$ are processed consecutively; see
  Fig.~\ref{fig:virtualEdgesInDummy}. We say that faces $f_2$ and
  $f_3$ create a \emph{facial pair}. Note that faces $f_1$, $f_2$ and
  $f_3$ are normal faces. In order to process a face, we require that
  it has a reference edge. The only case where this edge may not be
  defined, is for a face $f_4$ of a cycle that is processed after
  facial-pair $f_2$ and $f_3$. In this case we add a virtual edge as
  shown in Fig.~\ref{fig:virtualEdgesInDummy}.

  Consider a walk around the outer face of $G'$. An edge $e$ might be
  a planar edge of $G$, or a half-edge (incident to an out-dummy) or
  it is an edge that does not belong to $G$ and was added for caging a
  dummy-cut. We claim that there exists at least one edge $e=(s,s')$
  that is either planar or a half-edge. Indeed, consider the
  planarized graph $G_\mathrm{p}$ derived from $G$ and a leaf-component $C$ of
  its BC-tree decomposition. If $C$ is the root-component then it
  clearly consists only of planar and half-edges. Otherwise, $C$
  contains exactly one dummy-cut with degree two and at least two more
  vertices (the endpoints of the crossing edges of the
  dummy-cut). Since there are no other dummy-cuts and $C$ is
  biconnected, there exists a path between the two vertices that
  contains only planar edges and half-edges as claimed.

  Let $e=(s,s')$ be a planar or half-edge of $G'$. We subdivide $e$
  twice adding vertices $s_1$ and $s_2$. Then our reference edge for
  $G'$ will be edge $(s_1,s_2)$ where both $s_1$ and $s_2$ have degree
  2. We draw $(s_1,s_2)$ as a vertical segment and continue with the
  first face. Note that $s$ and $s'$ are diagonally aligned and edge
  (or half-edge) $(s,s')$ uses the west port of $s$ and south port of
  $s'$. We replace the three segments used for $(s,s')$ with a
  3-quarter arc, i.e. with curve complexity one as shown in
  Fig.~\ref{fig:smooth_first_edge}.

  \begin{figure}[tb]
    \begin{minipage}[b]{.6\textwidth}
      \centering
      \begin{subfigure}[b]{.34\linewidth}
        \centering
        \includegraphics[page=1, width=\textwidth]{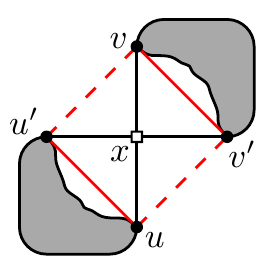}
        \caption{}
        \label{fig:virtualEdgesDummyCut}
      \end{subfigure}
      \hfil
      \begin{subfigure}[b]{.34\linewidth}
        \centering
        \includegraphics[page=2, width=\textwidth]{outer-smooth}
        \caption{}
        \label{fig:virtualEdgesInDummy}
      \end{subfigure}
      \caption{Introduction of virtual edges for dummy-cuts and
        in-dummies.  }
      \label{fig:virtualEdges}
    \end{minipage}
    \hfill
    \begin{minipage}[b]{.37\textwidth}
      \centering
      \includegraphics[page=3]{smooth_start}
      \caption{Starting the algorithm. ~~}
      \label{fig:smooth_first_edge}
    \end{minipage}
  \end{figure}

  In the following we describe how to draw $G$ based on the order
  defined on the faces of $G'$. We proceed by adding one face of $G'$
  at a time. We make sure that virtual edges are not present in the
  final drawing and that when cut-faces are processed we draw the two
  crossing edges instead.  Note that we deviate from the initial
  algorithm only if a face of $G'$ contains either in-dummy vertices
  or is a cut-face. In order to process in-dummies, we may also use
  convex quarter-arcs as we shall shortly see. At each step, we use
  one curve per edge, except for three special
  cases where we use edges of complexity two:
  \begin{inparaenum}[(i)]
  \item the case of crossing edges that are drawn when processing a
    cut-face, 
  \item planar edges of the outer face (which may be virtual) that
    cage in-dummy vertices, and
  \item planar edges of a triangular face (which may be virtual) with its
    reference edge drawn as a convex quarter arc.
  \end{inparaenum}
  
	Let $f$ be the next face to process. We distinguish three cases
  depending on whether $f$ is a normal face and does not belong to a
  facial-pair, is a normal face and belongs to a facial-pair, or is a
  cut-face.  During the process we respect the following invariants:
  \begin{enumerate}[\text{I}.1]
  \item \label{P.1} If an edge is incident to an out-dummy, it is
    always drawn with curve complexity~one, except for the case where
    it is a side-arc not incident to the north or east port of the
    out-dummy, where it may be drawn with curve complexity~two.
  \item \label{P.2} If a virtual edge is a reference for a normal
    face, it is always drawn as a quarter-arc, either convex or
    concave, and its endpoints are diagonally aligned.
  \item \label{P.3} If the reference edge for a cut-face is
    non-virtual, then it is drawn with curve complexity one as a convex
    or concave quarter-arc, or as horizontal or vertical segment.
  \item \label{P.4} If a virtual side-arc is a reference for a
    cut-face, it is always drawn with curve complexity~two.
  \end{enumerate}

  The last invariant, namely I.\ref{P.4}, implies that we need to
  alter the traditional drawing of Alam et al~\cite{ABKKKW14} in the
  cases where side-arcs are used. The corresponding drawings are given
  in Fig.~\ref{fig:alamEtAlOuterSmoothSide}.

  \begin{figure}[htb]
    \centering
    \begin{subfigure}[b]{.2\linewidth}
      \centering
      \includegraphics[page=1,width=\textwidth]{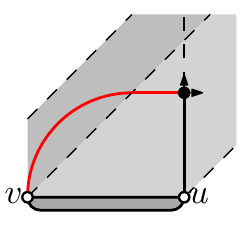}
      \caption{}
      \label{fig:alamEtAlOuterSmoothSide1a}
    \end{subfigure}
    \qquad
    \begin{subfigure}[b]{.2\linewidth}
      \centering
      \includegraphics[page=3, width=\textwidth]{alam-outer-smooth-side}
      \caption{}
      \label{fig:alamEtAlOuterSmoothSide2a}
    \end{subfigure}
    \qquad
    \begin{subfigure}[b]{.2\linewidth}
      \centering
      \includegraphics[page=2, width=\textwidth]{alam-outer-smooth-side}
      \caption{}
      \label{fig:alamEtAlOuterSmoothSide1b}
    \end{subfigure}
    \qquad
    \begin{subfigure}[b]{.2\linewidth}
      \centering
      \includegraphics[page=4, width=\textwidth]{alam-outer-smooth-side}
      \caption{}
      \label{fig:alamEtAlOuterSmoothSide2b}
    \end{subfigure}
    \caption{Using side-arcs with curve complexity two when (a)--(b)
      drawing singleton faces and (c)--(d) drawing chains. Red edges
      indicate virtual reference edges.}
    \label{fig:alamEtAlOuterSmoothSide}
  \end{figure}

  In addition to the invariants, we ensure that the following
  property holds throughout the process:

  \begin{enumerate}[P.1]
  \item \label{P.2a} For an edge that is drawn as a convex quarter-arc,
    it holds that its two endpoints are not dummy vertices and have
    remaining degree at most one.
  \end{enumerate}

  \paragraph{$f$ is a normal face and does not belong to a facial pair.}

  Suppose first that we encounter a normal face $f$ where the reference
  edge is a virtual edge $(u,v)$. By Invariant~I.\ref{P.2} $(u,v)$ is
  drawn as a quarter-arc. We draw $f$ by replacing the virtual edge as
  demonstrated in Fig.~\ref{fig:normal1} if $(u,v)$ is concave or as in
  Fig.~\ref{fig:normal2} if it is convex. Note that all edges are drawn
  as quarter-arcs, horizontal or vertical segments.  On the other hand,
  if the reference edge $(u,v)$ is not virtual, then the only case we
  need to pay extra attention is when $(u,v)$ is drawn as a convex
  quarter arc. If $f$ is a singleton we use the configuration of
  Fig.~\ref{fig:normal3}, where we draw the two edges with curve
  complexity two. If $f$ has more vertices, we use the placement of
  Figs.~\ref{fig:normal4}--\ref{fig:normal6}, depending on whether the
  two side-arcs are virtual reference edges for cut-faces or not (hence,
  we preserve Invariants~I.\ref{P.3}
  and~I.\ref{P.4}). Property~P.\ref{P.2a} and Invariant~I.\ref{P.2} are
  satisfied since we do not use convex quarter-arcs and there are no
  virtual reference edges for normal faces. The only case that could
  violate Invariant~I.\ref{P.1} is the case of adding a singleton face
  as in Fig.~\ref{fig:normal3}. However, the two vertices incident to
  the reference edge are not dummy vertices (by
  Property~P.\ref{P.2a}). Also, after inserting the singleton face they
  have remaining degree zero, therefore, the third vertex cannot be an
  out-dummy either.

  \begin{figure}[tb]
    \centering
    \begin{subfigure}[b]{.20\linewidth}
      \centering
      \includegraphics[page=1, width=\textwidth]{smooth_normal}
      \caption{}
      \label{fig:normal1}
    \end{subfigure}
    \qquad
    \begin{subfigure}[b]{.20\linewidth}
      \centering
      \includegraphics[page=2, width=\textwidth]{smooth_normal}
      \caption{}
      \label{fig:normal2}
    \end{subfigure}
    \qquad
    \begin{subfigure}[b]{.20\linewidth}
      \centering
      \includegraphics[page=3, width=\textwidth]{smooth_normal}
      \caption{}
      \label{fig:normal3}
    \end{subfigure}

    \begin{subfigure}[b]{.20\linewidth}
      \centering
      \includegraphics[page=4, width=\textwidth]{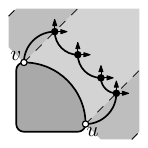}
      \caption{}
      \label{fig:normal4}
    \end{subfigure}
    \qquad
    \begin{subfigure}[b]{.20\linewidth}
      \centering
      \includegraphics[page=5, width=\textwidth]{smooth_normal}
      \caption{}
      \label{fig:normal5}
    \end{subfigure}
    \qquad
    \begin{subfigure}[b]{.20\linewidth}
      \centering
      \includegraphics[page=6, width=\textwidth]{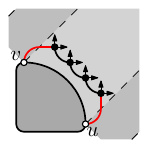}
      \caption{}
      \label{fig:normal6}
    \end{subfigure}
    \caption{Processing a normal face whose reference edge is
      (a)--(b)~virtual, and (c)--(f)~non virtual.}
    \label{fig:normal}
  \end{figure}

  \paragraph{$f$ is a normal face and belongs to a facial pair.}
  This case appears only when there is a cycle of four faces, say
  $f_1,f_2,f_3,f_4$, with $x$ as their common in-dummy vertex. Face
  $f_1$ is a normal face and has already been drawn and $f$ is either
  $f_2$ or $f_3$. Since faces $f_2$ and $f_3$ are consecutive, we show
  how to draw them simultaneously, so that an appropriate reference edge
  can be defined for $f_4$ that will be drawn at a later step.  Observe
  that $f_2$ and $f_3$ are normal faces and can be drawn as explained in
  the previous case.  Let $(u,v)$ and $(u',v')$ be the two crossing
  edges of $G$, such that $f_1$ contains vertices $u$, $u'$ and $x$,
  $f_2$ contains $x$, $u^\prime$ and $v$, and $f_3$ contains $x$, $u$
  and $v'$.

  Also note that when $f_1$ is drawn, vertex $x$ has north and east
  ports available, so that edges $(x,v)$ and $(x,v')$ are drawn as
  vertical and horizontal segments for $f_2$ and $f_3$ respectively.
  For edge $(v,v^\prime)$ to become a reference edge for face $f_4$, we
  want vertices $v,v^\prime$ to be diagonally aligned (to preserve
  Invariant~I.\ref{P.2}).  In order to do this we consider the following
  cases depending on the placement of vertices $u$, $x$, and $u'$, as
  produced by the drawing of face $f_1$. Note that edges $(u,x)$ and
  $(x,u')$ are always drawn with curve complexity one. Furthermore, they
  can be drawn neither as side-arcs of face $f_1$ (since they do not
  belong to the outer face), nor as convex quarter-arcs by
  Property~P.\ref{P.2a}.

  \begin{itemize}
  \item Edges $(u,x)$ and $(x,u')$ are drawn as consecutive concave
    quarter arcs (see Fig.~\ref{fig:facial_pair1}). Those arcs have the
    same size. Note that we are able to modify the sizes of the two arcs
    by moving $x$ along the diagonal $(u,u')$ and without affecting the
    properties of the drawing. When moving $x$ from $u$ to $u'$, the
    vertical segment used for $(x,v)$ increases and the horizontal
    segment for $(x,v')$ decreases. Hence, we can find a placement for
    $x$ so that the two segments have same length; see
    Fig.~\ref{fig:facial_pair1a}.
  \item Edge $(u',x)$ is drawn as horizontal segment and $(x,u)$ as a
    concave quarter arc (see Figs.~\ref{fig:facial_pair2} and
    ~\ref{fig:facial_pair4}). In the case where the length of $(x,v)$ is
    smaller than the length of $(x,v')$ (refer to
    Fig.~\ref{fig:facial_pair2}), we move $x$ towards $v'$ and redraw
    $(x,u)$ with curve complexity two as shown in
    Fig.~\ref{fig:facial_pair3}. The side-arc of $f_2$ is also redrawn
    with curve complexity two. In the other case, if the length of
    $(x,v)$ is greater than the length of $(x,v')$ as in
    Fig.~\ref{fig:facial_pair4}, we redraw $f_2$ by reducing the size of all new edges including  $(x,v)$, and redrawing the side-arc of $f_2$ with curve complexity
    two. The corresponding drawing is shown in
    Fig.~\ref{fig:facial_pair5}.
  \item The case of edge $(u',x)$ drawn as concave quarter arc and
    $(x,u)$ as a vertical segment is symmetric to the previous one.
  \item Edge $(u',x)$ is drawn as horizontal segment and $(x,u)$ as
    vertical segment (see Fig.~\ref{fig:facial_pair6}). We shorten the
    longer segment and redraw the side-edge of its face with edge
    complexity two as shown in Fig.~\ref{fig:facial_pair7}.
  \end{itemize}

  Note that for the in-dummy vertex each incident half-edge is drawn
  with curve complexity~one, except for the case shown in
  Fig.~\ref{fig:facial_pair3}. In this special case, however, the
  corresponding crossing edges of $G$ are still of curve complexity~two.

  \begin{figure}[tb]
    \centering

    \begin{subfigure}[b]{.2\linewidth}
      \centering
      \includegraphics[page=1, width=\textwidth]{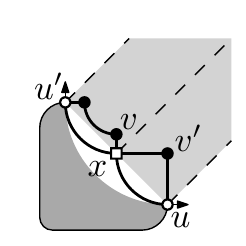}
      \caption{}
      \label{fig:facial_pair1}
    \end{subfigure}
    \qquad
    \begin{subfigure}[b]{.2\linewidth}
      \centering
      \includegraphics[page=3, width=\textwidth]{facial_pair}
      \caption{}
      \label{fig:facial_pair2}
    \end{subfigure}
    \qquad
    \begin{subfigure}[b]{.2\linewidth}
      \centering
      \includegraphics[page=5, width=\textwidth]{facial_pair}
      \caption{}
      \label{fig:facial_pair4}
    \end{subfigure}
    \qquad
    \begin{subfigure}[b]{.2\linewidth}
      \centering
      \includegraphics[page=7, width=\textwidth]{facial_pair}
      \caption{}
      \label{fig:facial_pair6}
    \end{subfigure}
    \qquad
    \begin{subfigure}[b]{.2\linewidth}
      \centering
      \includegraphics[page=2, width=\textwidth]{facial_pair}
      \caption{}
      \label{fig:facial_pair1a}
    \end{subfigure}
    \qquad
    \begin{subfigure}[b]{.2\linewidth}
      \centering
      \includegraphics[page=4, width=\textwidth]{facial_pair}
      \caption{}
      \label{fig:facial_pair3}
    \end{subfigure}
    \qquad
    \begin{subfigure}[b]{.2\linewidth}
      \centering
      \includegraphics[page=6, width=\textwidth]{facial_pair}
      \caption{}
      \label{fig:facial_pair5}
    \end{subfigure}
    \qquad
    \begin{subfigure}[b]{.2\linewidth}
      \centering
      \includegraphics[page=8, width=\textwidth]{facial_pair}
      \caption{}
      \label{fig:facial_pair7}
    \end{subfigure}
    \caption{Drawing faces that belong to a facial pair. Blue edges
      may be virtual or non-virtual.}
    \label{fig:facial_pair}
  \end{figure}

  After aligning vertices $v$ and $v^\prime$ we draw reference edge
  $(v,v^\prime)$ as a convex quarter arc; see
  Figs.~\ref{fig:facial_pair1a}--\ref{fig:facial_pair7}. It might be
  the case that this is a (virtual) reference edge, but then it
  satisfies Invariants~I.\ref{P.2} and~I.\ref{P.3}. Since $v$ and
  $v^\prime$ are incident to dummy $x$, Property~P.\ref{P.2a} is
  preserved. We can argue that the remaining invariants are also
  satisfied in a similar way as we did in the case where $f$ did not
  belong to a facial pair. Invariant~I.\ref{P.1} is maintained by the initial drawing of $f_2$ and $f_3$ as normal faces using the additional configurations of Fig.~\ref{fig:alamEtAlOuterSmoothSide}.
	

  \paragraph{$f$ is a cut-face.}

  For any cut-face of $G'$ there exist two virtual edges (see $(u,v')$
  and $(u',v)$ in Fig.~\ref{fig:virtualEdgesDummyCut})
  on the outer face of $G'$ used solely for preserving
  biconnectivity. We ignore these two edges when processing $f$ and draw
  the pair of crossing edges $(u,v)$ and $(u',v')$ instead. We
  distinguish two cases, depending on whether edge $(u,u')$ exists in
  $G$ or is a virtual edge. In the first case, we use one of the
  configurations shown in Fig.~\ref{fig:smooth_cut}, depending on how
  $(u,u')$ is drawn. Note, that $(u,u')$ is not on the outer face of
  $G'$, and therefore is not a side-arc of its face. For the other case,
  where $(u,u')$ is a virtual edge, we use the configurations of
  Fig.~\ref{fig:smooth_cut_virtual} for all possible drawings of
  $(u,v)$. Observe, that if virtual edge $(u,u')$ is a side-arc, it is
  drawn with curve complexity~two by Invariant~I.\ref{P.4}.  Since apart
  from crossing edges $(u,v)$ and $(u',v')$, we only add new reference
  edge $(v,v')$ which is drawn as a concave quarter circle, we preserve
  Invariants~I.\ref{P.1}--I.\ref{P.4}.

  \begin{figure}[tb]
    \centering
    \begin{subfigure}[b]{.2\linewidth}
      \centering
      \includegraphics[page=1, width=\textwidth]{smooth_cut}
      \caption{}
      \label{fig:smooth_cut1}
    \end{subfigure}
    \qquad
    \begin{subfigure}[b]{.2\linewidth}
      \centering
      \includegraphics[page=2, width=\textwidth]{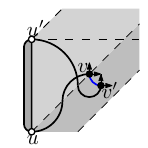}
      \caption{}
      \label{fig:smooth_cut2}
    \end{subfigure}
    \qquad
    \begin{subfigure}[b]{.2\linewidth}
      \centering
      \includegraphics[page=3, width=\textwidth]{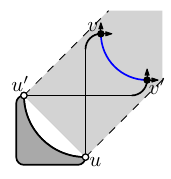}
      \caption{}
      \label{fig:smooth_cut3}
    \end{subfigure}
    \qquad
    \begin{subfigure}[b]{.2\linewidth}
      \centering
      \includegraphics[page=4, width=\textwidth]{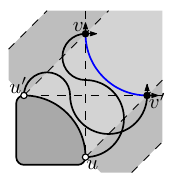}
      \caption{}
      \label{fig:smooth_cut4}
    \end{subfigure}

    \caption{Drawing a cut-face where the reference edge $(u,v)$ exists
      and is drawn as (a)~an horizontal segment, (b)~a vertical segment,
      (c)~a concave quarter-arc, or (d)~a convex quarter-arc.}
    \label{fig:smooth_cut}
  \end{figure}

  \begin{figure}[tb]
    \centering
    \begin{subfigure}[b]{.2\linewidth}
      \centering
      \includegraphics[page=1,
      width=\textwidth]{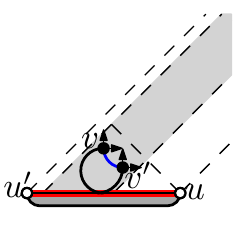}
      \caption{}
      \label{fig:smooth_cut_virtual1}
    \end{subfigure}
    \qquad
    \begin{subfigure}[b]{.2\linewidth}
      \centering
      \includegraphics[page=2,
      width=\textwidth]{smooth_cut_virtual}
      \caption{}
      \label{fig:smooth_cut_virtual2}
    \end{subfigure}
    \qquad
    \begin{subfigure}[b]{.2\linewidth}
      \centering
      \includegraphics[page=3,
      width=\textwidth]{smooth_cut_virtual}
      \caption{}
      \label{fig:smooth_cut_virtual3}
    \end{subfigure}

    \begin{subfigure}[b]{.2\linewidth}
      \centering
      \includegraphics[page=4,
      width=\textwidth]{smooth_cut_virtual}
      \caption{}
      \label{fig:smooth_cut_virtual4}
    \end{subfigure}
    \qquad
    \begin{subfigure}[b]{.2\linewidth}
      \centering
      \includegraphics[page=5,
      width=\textwidth]{smooth_cut_virtual}
      \caption{}
      \label{fig:smooth_cut_virtual5}
    \end{subfigure}
    \qquad
    \begin{subfigure}[b]{.2\linewidth}
      \centering
      \includegraphics[page=6,
      width=\textwidth]{smooth_cut_virtual}
      \caption{}
      \label{fig:smooth_cut_virtual6}
    \end{subfigure}

    \caption{Drawing a cut-face where the reference edge $(u,v)$
      is virtual and is drawn as (a)~an horizontal segment, (b)~a
      vertical segment, (c)~a concave quarter-arc, (d)~a convex
      quarter-arc.  (e)--(f)~a side edge.}
    \label{fig:smooth_cut_virtual}
  \end{figure}

  \paragraph{Correctness.}

  The produced layout of $G'$ is a smooth layout for our initial graph
  $G$, if we consider in-dummy and out-dummy vertices as crossing
  points. We claim that each edge has curve complexity at most two. This
  is true for all planar edges and crossing edges which create
  cut-faces.  Also, crossing edges corresponding to in-dummies have edge
  complexity two as already claimed. The only case that remains to
  consider are crossing edges corresponding to out-dummies. Recall that
  by Invariant~I.\ref{P.1} a half-edge incident to an out-dummy is drawn
  with curve complexity~two only if it is a side-arc incident to either
  south or west port of the out-dummy. This case only occurs when
  processing facial pairs (refer to Fig.~\ref{fig:facial_pair}). Assume
  that half-edge $e$ is incident to the west port of out-dummy $x$. Then
  $e$ has a horizontal segment incident to $x$. To achieve edge
  complexity~two for the crossing edge of $G$, we have to argue that the
  half-edge $e'$ incident to $x$'s east port is drawn as a horizontal
  segment. By Invariant~I.\ref{P.1}, $e'$ is drawn with edge
  complexity~one, and by Property~P.\ref{P.2a}, $e'$ is not a convex
  quarter-arc. So, assume that $e'$ was drawn as a side-arc. This
  implies that both $e$ and $e'$ are on the outer face of $G'$. A clear
  contradiction since an out-dummy has exactly two consecutive edges on
  the outer face. We conclude that $e'$ has to be drawn as a horizontal
  segment. The case where $e$ is incident to the south-port is treated
  similarly.

  To ensure that our construction preserves 1-planarity, we use a
  similar argument as in~\cite{ABKKKW14}. In particular, for every reference
  edge $(u,v)$, we define a diagonal semi-strip $L_{u,v}$ delimited by
  lines of slope $+1$ through $u$ and $v$; highlighted in light-gray in
  Figs.~\ref{fig:alamEtAlOuterSmoothSide}--\ref{fig:smooth_cut_virtual}.
  Note that the two diagonal lines are not part of $L_{u,v}$. In
  addition, we may extend $L_{u,v}$ by two semi-strips $L_{u,v}^t$ and
  $L_{u,v}^b$ of $0.42$ times\footnote{More precisely: The widths of
    $L_{u,v}^t$ and $L_{u,v}^b$ must be at least $(\sqrt{2}-1)$ times
    the width of $L_{u,v}$ to fully contain a side-arc.} the width of
  $L_{u,v}$ surrounding $L_{u,v}$ from above and below, respectively;
  see dark-gray semi-strips in
  Figs.~\ref{fig:alamEtAlOuterSmoothSide}--\ref{fig:smooth_cut_virtual}.
  However, we only extend $L_{u,v}$ if it is ensured that neighboring
  light-gray semi-strips are empty by the degree restriction (in
  particular, this is the case when the construction uses the north port
  of $v$ or east port of $u$, respectively). The drawing for the
  subgraph emerging from reference edge $(u,v)$ then is completely
  contained in $L_{u,v}$, $L_{u,v}^t$ (if it is introduced) and
  $L_{u,v}^b$ (if it is introduced). Observe that if a new reference
  edge $(u',v')$ is created in $L_{u,v}^t$ or $L_{u,v}^b$, we can choose
  a suitably small length of $(u',v')$ such that the entire subgraph
  emerging from $(u',v')$ remains inside $L_{u,v}^t$ or
  $L_{u,v}^b$. Also note that if we draw a cut-face with virtual reference edge $(u,u')$ using one of the configurations of 
  Figs.~\ref{fig:smooth_cut_virtual1}--\ref{fig:smooth_cut_virtual4}, we
  restrict ourselves to a smaller semi-strip of width $0.16$ times the normal
  width of $L_{u,u'}$ for drawing the component with reference edge
  $(v,v')$ (not up to scale in the figures) so that dark gray semi-strips
  defined for $u$ or $u^\prime$ do not overlap $L_{u,u'}$.
\end{proof}}{}

\end{document}